\documentclass[journal,twocolumn]{IEEEtran}

\usepackage[cmex10]{amsmath}
\usepackage{mathrsfs}
\usepackage{cite}
\usepackage{todonotes}
\usepackage{stfloats}

\usepackage{amssymb,amsthm}
\usepackage{tikz}
\usetikzlibrary{calc}

\newtheorem{theorem}{Theorem}
\newtheorem{corollary}{Corollary}
\newtheorem{proposition}{Proposition}
\newtheorem{lemma}{Lemma}

\theoremstyle{definition}

\newtheorem{policy}{Policy}
 
\theoremstyle{remark}
\newtheorem*{remark}{Remark}

\newcounter{mytempeqncnt}

\begin{document}
\allowdisplaybreaks

\title{Online Power Control for Block i.i.d.\ Energy Harvesting Channels}

\author{Dor Shaviv and Ayfer \"{O}zg\"{u}r
\thanks{This work was supported in part by a Robert Bosch Stanford Graduate Fellowship, in part by the National Science Foundation under Grant CCF-1618278, and in part by the Center for Science of Information, an NSF Science and Technology Center, under Grant CCF-0939370.
This work was presented in part at the 2017 IEEE Wireless Communications and Networking Conference (WCNC)~\cite{BlockiidBernoulli} and submitted to 2017 IEEE GLOBECOM~\cite{BlockiidGLOBECOM}.
}
\thanks{The authors are with the Department of Electrical Engineering, Stanford University, Stanford, CA 94305 USA (e-mail: shaviv@stanford.edu; aozgur@stanford.edu).}
}

\maketitle

\begin{abstract}
We study the problem of online power control for energy harvesting communication nodes with random energy arrivals and a finite battery. We assume a block i.i.d. stochastic model for the energy arrivals, in which the energy arrivals are constant for a fixed duration $T$, but are independent across different blocks, drawn from an arbitrary distribution. This model serves as a simple approximation to a random process with coherence time $T$. We propose a simple online power control policy, and prove that its performance gap to the optimal throughput is bounded by a constant which is independent of the parameters of the problem. This also yields a simple formula for the approximately optimal long-term average throughput, which sheds some light on the qualitative behavior of the throughput and how it depends on the coherence time of the energy arrival process. Our results show that, perhaps counter-intuitively, for a fixed mean energy arrival rate the throughput decreases with increasing coherence time $T$ of the energy arrival process. In particular, the battery size needed to approach the AWGN capacity of the channel increases linearly with the coherence time of the process. Finally, we show that our results can provide an approximation to the information-theoretic capacity of the same channel.
\end{abstract}


\IEEEpeerreviewmaketitle

\section{Introduction}

Recent advances in energy harvesting technologies enable wireless devices to harvest the energy they need for communication from the natural resources in their environment. This development opens the exciting possibility to build wireless networks that are self-powered, self-sustainable and which have lifetimes limited by their hardware and not the size of their batteries.

Communication with such wireless devices requires the design of good power control policies that can maximize throughput under random energy availability. 
In particular, available energy should not be consumed too fast, or transmission can be interrupted in the future due to an energy outage; on the other hand, if the energy consumption is too slow, it can result in the wasting of the harvested energy and missed recharging opportunities in the future due to an overflow in the battery capacity.
This problem has received significant interest in the recent literature
\cite{YangUlukus2012,TutuncuogluYener2012,Ozeletal2011,
DP0,DP1,DP3,ho2012optimal,
Mitran,online_infiniteB2,online_infiniteB3,online_infiniteB4,info2013,amirnavaei2015online,
xu2014throughput,Kazerouni2015,
DongOzgur2014,DongFarniaOzgur2015,NearOptimal1,NearOptimal2,
HuseyinWiOpt,HuseyinISIT,
BakninaUlukusISIT2016MAC,BakninaUlukusISIT2016Broadcast,BakninaUlukusJSAC2016,VaranYenerWCNC2017}.
In the offline case, when future energy arrivals are known ahead of time, the problem has an explicit solution \cite{YangUlukus2012,Ozeletal2011,TutuncuogluYener2012}.
The optimal policy keeps energy consumption as constant as possible over time while ensuring no energy wasting due to an overflow in the battery capacity. The more interesting case is the online scenario where future energy arrivals are random and unknown. When the energy arrivals are i.i.d., the problem can be modeled as a Markov Decision Process (MDP) and solved numerically using dynamic programming \cite{DP0,DP1,DP3,ho2012optimal}. However, this numerical approach can be computationally demanding and does not provide insight into the structure of the optimal online power control policy and the qualitative behavior of the resultant throughput, namely how it varies with the parameters of the problem. This kind of insight can be critical for design considerations, such as choosing the size of the battery to employ at the transmitter. More recently in \cite{NearOptimal1}, we  developed a simple online policy, which provably achieves a near-optimal throughput for any distribution of the energy arrivals (see also precursory work in \cite{DongOzgur2014,DongFarniaOzgur2015} and extensions in 
\cite{HuseyinWiOpt,HuseyinISIT,BakninaUlukusISIT2016MAC,BakninaUlukusISIT2016Broadcast,BakninaUlukusJSAC2016,VaranYenerWCNC2017}).
The gap between the throughput achieved by this scheme and the optimal throughput can be explicitly bounded by a constant independent of the distribution of the energy arrivals and any of the problem parameters.
This leads to a simple approximation for the optimal throughput, which sheds some light on the qualitative behavior of the optimal throughput and its dependence on major problem parameters.

All of the above solutions, including the MDP approach, are applicable only when the energy arrival process is i.i.d., and therefore the next energy arrival at each time instant is impossible to predict.
However, most natural energy harvesting processes, such as solar energy or wind energy, are far from i.i.d.\ and are highly correlated over time. For processes of this type, an i.i.d.\ model is very far from the actual behavior of the process.
The research on optimal online power control for non-i.i.d.\ processes with finite battery size is very scarce.
For example, \cite{online_infiniteB4} proposes a simple policy for general stationary ergodic arrival processes which becomes asymptotically optimal as the battery size tends to infinity, however this strategy can be arbitrarily away from optimality at finite battery size.
For a finite battery, \cite{MaoHassibiarXiv} studies the information-theoretic capacity of a model with a general Markov arrival process, and provides upper and lower bounds on capacity. However, these bounds can be arbitrarily away from optimality, and moreover they do not provide any qualitative understanding of the actual capacity of the system.

In this work, we consider energy arrivals processes which follow a block i.i.d.\ model. This means that the energy arrivals remain constant for a fixed period of time, say $T$ time slots, and then change to an independent realization for the next $T$ time slots.
This can model, for example, a solar panel which harvests energy from the sun, and the appearance of clouds can change randomly and block certain amounts of sunshine for a certain period of time. This process can be approximated by a block i.i.d.\ model.
Additionally, this is a good model for a device which harvests RF energy from other transmitting devices in its environment. Such transmitting devices typically transmit continuously for certain periods of time and are silent for the remaining periods (as in TDMA, for example), which warrants a block i.i.d.\ model. Note that block i.i.d.\ models have been popularly used in wireless communication to capture correlations in the channel fading process by a simple model. In this case, $T$ is called the coherence time of the channel, which corresponds to the time duration over which the channel remains approximately constant \cite{Davidsbook}. Analogously, we refer to $T$ as the coherence time of the energy arrival process in this paper.

We propose a simple policy and establish its near-optimality for this block i.i.d.\ model. This policy combines features of the optimal offline~\cite{TutuncuogluYener2012} and approximately-optimal online~\cite{NearOptimal1} strategies for the i.i.d. ($T=1$) model. Since in the beginning of each block the future energy arrivals are known for a duration of $T$ channel uses, energy allocations for the entire block can be decided on ahead of time, akin to the offline setting. In particular, power allocation inside each block is constant, as implied by the optimal offline strategy, and ensures that energy is not wasted due to an overflow in the battery capacity. On the other hand, the energy arrivals are i.i.d.\ across different blocks and the situation across blocks is akin to the online setting. In particular, between different blocks, the policy resembles the \emph{Fixed Fraction Policy} of~\cite{NearOptimal1}, where a constant fraction of the currently available energy in the battery is allocated to the channel. However, achieving the optimal throughput within a constant gap requires a non-trivial combination of these two schemes. 

In the same spirit of~\cite{NearOptimal1}, we develop a lower bound to the throughput achieved by our proposed policy by modifying the distribution of the energy arrivals. We do so in a way that, as we show, produces \emph{worse} throughput than the original distribution,\footnote{In~\cite{NearOptimal1} this modified distribution was simply a Bernoulli distribution; here the modification is slightly more involved.} and for which we can analytically evaluate the throughput. We then proceed to developing a nearly-tight upper bound on the optimal throughput achievable under the block i.i.d. model. The throughput achieved with an infinite battery, namely the AWGN capacity $\tfrac{1}{2}\log(1+\mu)$ where $\mu$ is the mean of the energy arrival rate, is always an upper bound on the throughput achievable with any finite battery size. This was the upper bound used in the i.i.d. case in \cite{NearOptimal1}. However, this upper bound turns out to be too loose in general for the block i.i.d. case;
indeed, we show that this upper bound is nearly-achievable (up to a bounded gap) only when the battery size is large enough, specifically when $\bar{B}\geq \mu+T(E_{\max}-\mu)$, where $E_{\max}$ is the maximal energy arrival. 
Note that for fixed $\mu$ and $E_{\max}$, as the coherence time $T$ of the energy arrival process increases, a larger battery is needed to approach the AWGN capacity.
This is somewhat counter-intuitive, since one may expect a large coherence time to increase the optimal throughput, as it results in larger lookahead. We show that when $\bar{B}<\mu+T(E_{\max}-\mu)$, the optimal throughput can be significantly smaller than the AWGN capacity. We finally show that the difference between the throughput achieved by our proposed strategy and our upper bound is bounded by $\tfrac{1}{2}\log e\approx 0.72$, regardless of the values of the problem parameters. 

While in this paper we mostly focus on the online power control problem for energy harvesting nodes, we show that this problem is central to understanding and achieving the information-theoretic capacity of this channel. Following the approach in \cite{ShavivNguyenOzgur2016}, which focused on an i.i.d. model for the energy arrival process, we show that the information-theoretic capacity of the channel can be approximated by the corresponding optimal online throughput also under a block i.i.d. model. The upper bound on the gap between the two performance metrics we develop in this paper depends on the entropy rate of the energy arrival process, which decreases with increasing coherence time $T$. It is also possible to bound the gap between these two performance metrics by the entropy rate of the online power control process (rather than the entropy rate of the energy harvesting process itself.) By modifying the online power control policy to have a constant entropy rate in the lines of \cite{ShavivNguyenOzgur2016}, we believe it is possible to show that the information-theoretic capacity and the optimal online throughput are indeed within a constant gap of each other, independent of the parameters of the problem. 

\section{System Model}
\label{sec:system_model}

We begin by introducing notation: Let $\mathbb{E}[\,\cdot\,]$ denote expectation. 
All logarithms are taken to base 2.
For a process $\{x_t\}_{t=1}^{\infty}$ with a block structure, it will be convenient to have special notation for the $j$-th slot in the $i$-th block:
\begin{align*}
x_j^{(i)}:=x_{(i-1)T+j},
&&j=1,\ldots,T,
&&i=1,2,\ldots
\end{align*}

We consider the discrete-time online power control problem for an energy harvesting transmitter communicating over an additive Gaussian channel. The transmitter is equipped with a battery of finite capacity $\bar{B}$, which is being continuously recharged by an exogenous energy harvesting process.
Let $E_t\in\mathcal{E}$ be the energy harvested at discrete time $t$.
We assume $E_t$ is a \emph{block i.i.d.} stochastic process, with block duration $T$.
More precisely, let $\{E^{(i)}\}_{i=1}^{\infty}$ be an i.i.d. random process, where $E^{(i)}\in\mathcal{E}$ is a nonnegative random variable (RV) drawn from a set $\mathcal{E}$ with marginal distribution $P_E$.
To simplify the analysis, we assume $P_E$ is a discrete distribution and $\mathcal{E}$ is a finite set, however our results hold in more generality for arbitrary discrete or continuous distributions.
Then the process $E_t$ is given by
\begin{align*}
E_t=E^{(i)},
&&t=(i-1)T+1,\ldots,iT,
&&i=1,2,\ldots
\end{align*}
We assume $E^{(i)} > 0$ with positive probability; otherwise $E^{(i)}=0$ w.p. 1 and the problem is degenerate.

A power control \emph{policy} for an energy harvesting system is a sequence of mappings from energy arrivals to a non-negative number, which will denote a level of instantaneous power.
In this work, we will focus on \emph{online policies};
an {online policy} $\mathbf{g}=\{g_t\}_{t=1}^{\infty}$ is a sequence of mappings 
$g_t:\mathcal{E}^t\to\mathbb{R}_+$, $t=1,2,\ldots$, such that the instantaneous power at time $t$ is $g_t(E_1,E_2,\ldots,E_t)$. In words, the power allocation at time $t$ can depend only on the realizations of the energy arrival process up to time $t$ (and not future realizations of the random process), although the probabilistic model for the energy arrival process, i.e. the fact that it is block i.i.d. with coherence time $T$ and distribution $P_E$, is known ahead of time. 
By allocating power $g_t$ at time $t$, we get an instantaneous rate equal to the AWGN capacity, i.e., 
\begin{equation}
C(g_t):=\tfrac{1}{2}\log(1+ g_t).\label{eq:AWGNcap}
\end{equation}

Let $b_t$ be the amount of energy available in the battery at the beginning of time slot $t$. 
An \emph{admissible policy} $\mathbf{g}$ is such that satisfies the following constraints for every possible harvesting sequence $\{E_t\}_{t=1}^{\infty}$:
\begin{align}
&0\leq g_t\leq b_t&&,t=1,2,\ldots,\label{eq:power}\\
&b_t=\min\{b_{t-1}-g_{t-1}+E_t,\bar{B}\}&&,t=2,3,\ldots,\label{eq:battery}
\end{align}
where we assume $b_1=\bar{B}$ without loss of generality.


For a given policy $\mathbf{g}$, we define the $N$-horizon expected total throughput to be
$\mathscr{T}_N(\mathbf{g})=\mathbb{E}[
\sum_{t=1}^{N}C(g_t)]$,
where the expectation is over the energy arrivals $E_1,\ldots,E_N$.
The long-term average throughput of the same policy is defined as
\begin{equation}
\mathscr{T}(\mathbf{g})=\liminf_{N\to\infty}\tfrac{1}{N}\mathscr{T}_N(\mathbf{g}).
\label{eq:def_long_term_throughput}
\end{equation}
Our goal is to characterize the optimal online power control policy and the resultant optimal long-term average throughput:
\begin{equation}
\Theta=
\sup_{\mathbf{g}\text{ admissible}}\mathscr{T}(\mathbf{g}).
\label{eq:def_optimal_throughput}
\end{equation}

\section{Preliminary Discussion}
\label{sec:preliminary_discussion}

\subsection{Background}
\label{subsec:background}

The optimal offline power control policy has been explicitly characterized in~\cite{YangUlukus2012,Ozeletal2011,TutuncuogluYener2012}, in which the energy arrival sequence $\{E_t\}_{t=1}^{\infty}$ is assumed to be known ahead of time.
Additionally, in~\cite{NearOptimal1} we develop a near-optimal online power control policy for the case of i.i.d. energy arrivals, and provide an approximate expression for the resultant long-term average throughput, with a bounded gap to optimality. In particular,  
the Fixed Fraction Policy of \cite{NearOptimal1} allocates a fixed fraction $q$ of the currently available energy at each channel use. 
More precisely, let $q\triangleq\frac{\mathbb{E}[\min(E_t,\bar{B})]}{\bar{B}}$. Note that $0\leq q\leq 1$. Then
\[
g_t=q b_t,\qquad t=1,2,\ldots
\]
The main result of \cite{NearOptimal1} is to establish the optimality of this online strategy within a constant additive gap for any i.i.d.\ process ($T=1$).
\begin{theorem}[Theorem 2 in \cite{NearOptimal1}]\label{thm:onlinePC}
Let $E_t$ be an i.i.d.\ non-negative process, and let $\mathbf{g}$ be the Fixed Fraction Policy.
Then the throughput achieved by $\mathbf{g}$ is bounded by
\[
\mathscr{T}(\mathbf{g})\geq C(\mathbb{E}[\min(E_t,\bar{B})])-\frac{1}{2}\log e,
\]
where $C(\cdot)$ is the AWGN capacity defined in \eqref{eq:AWGNcap}.
\end{theorem}
Note that the AWGN capacity $C(\mathbb{E}[\min(E_t,\bar{B})])$ is an upper bound on the achievable throughput for any distribution (see \cite[Prop. 2]{NearOptimal1}).
Observe that whenever there is an energy arrival larger than the battery size, $E_t>\bar{B}$, the battery will be completely recharged to $\bar{B}$ and the remaining energy is discarded as per \eqref{eq:battery}. Hence, effectively, this is as if an energy arrival of $E_t=\bar{B}$ occurred. We can therefore replace the energy arrival process with $\min(E_t,\bar{B})$, and $\mathbb{E}[\min(E_t,\bar{B})]$ is the mean energy harvested by the transmitter.

There is an alternative way to view the quantity $\mathbb{E}[\min(E_t,\bar{B})]$, which will be useful in the sequel. Observe that whenever the event $\{E_t>\bar{B}\}$ occurs, the memory of the system, which is encapsulated in the state of the battery, is essentially erased. This induces a regenerative structure for the online decision process, and the behaviors of different \emph{epochs}---the periods between consecutive events $\{E_t>\bar{B}\}$---are statistically independent and identical.
Let $p=\Pr(E_t>\bar{B})$, and observe that the average length of an epoch is $\tau=1/p$.
The average energy available for transmission in a single epoch is given by $\varepsilon=\bar{B}+(\tfrac{1}{p}-1)\mathbb{E}[E_t|E_t\leq\bar{B}]$, because the battery is fully charged at the beginning of the epoch, and the average amount of energy harvested in each of the subsequent time slots is $\mathbb{E}[E_t|E_t\leq\bar{B}]$.
Therefore, the average energy per time slot which is available for transmission is given by
\begin{align*}
\frac{\varepsilon}{\tau}
&=p\bar{B}+(1-p)\mathbb{E}[E_t|E_t\leq\bar{B}]=\mathbb{E}[\min(E_t,\bar{B})].
\end{align*}

\subsection{Preliminary Results}
\label{subsec:preliminary_results}

For the block i.i.d.\ model considered in this paper,
it can be observed that the problem can be formulated as a Markov Decision Process (MDP), where each time step of the MDP corresponds to $T$ time slots of the original communication system.
Let $i$ denote the $i$-th step of this MDP. Then we define the state as the pair
$(b_{(i-1)T+1},E_{(i-1)T+1})=(b_1^{(i)},E^{(i)})$.
The action (or control) is the vector of power allocations for the entire block
$(g_{(i-1)T+1},\ldots,g_{iT})=(g_1^{(i)},\ldots,g_T^{(i)})$,
which must satisfy the energy constraints~\eqref{eq:power} and~\eqref{eq:battery}.
The disturbance is $E^{(i+1)}$, and the next state pair $(b_1^{(i+1)},E^{(i+1)})$ is given by
\begin{align}
b_1^{(i+1)}&=\min\{b_T^{(i)}-g_T^{(i)}+E^{(i+1)},\bar{B}\},
\end{align}
where the state variable $E^{(i+1)}$ is of course equal to the disturbance itself.
The stage reward is given by
$r_i=\tfrac{1}{T}\sum_{j=1}^{T}C(g_j^{(i)})$, and the goal is to optimize the expected long-term average reward per stage, given by $\liminf_{n\to\infty}\tfrac{1}{N}\sum_{i=1}^{N}\mathbb{E}[r_i]$.

In fact, this MDP can be further simplified.
First, it can be easily seen that, since the energy arrivals for the entire block are known ahead of time, it is suboptimal to have battery overflows inside the block (unless $E^{(i)}\geq\bar{B}$, in which case overflows are inevitable). That is, $b_j^{(i)}=b_{j-1}^{(i)}-g_{j-1}^{(i)}+E^{(i)}$ for $j=2,\ldots,T$.
Otherwise, if there is some $j$ such that $b_{j-1}^{(i)}-g_{j-1}^{(i)}+E^{(i)}>\bar{B}$, one can simply increase $g_{j-1}^{(i)}$, and consequently the reward, without affecting the state.
According to this observation, and by concavity of the logarithm, it follows that it is optimal to set $g_1^{(i)}=g_2^{(i)}=\ldots=g_{T-1}^{(i)}$.
Thus the control is reduced to the pair $(g_1^{(i)},g_T^{(i)})$ 
(note that in general $g_T^{(i)}$ is not equal to $g_T^{(i)}$).
This is made formal in the following lemma, which is proved in Appendix~\ref{sec:MDP_reduction}.

\begin{lemma}
The MDP defined previously is equivalent to the following MDP, with
state pair $(b_1^{(i)},E^{(i)})$, action pair $(g_1^{(i)},g_T^{(i)})$, and disturbance $E^{(i+1)}$.
The actions must satisfy the constraints 
\begin{equation}
0\leq g_1^{(i)}\leq \min\big(E^{(i)}-\tfrac{E^{(i)}-b_1^{(i)}}{T-1},\bar{B}\big),
\label{eq:g1_action_space}
\end{equation}
\begin{equation}
0\leq g_T^{(i)}\leq b_T^{(i)},
\label{eq:gT_action_space}
\end{equation}
where $b_T^{(i)}=\min\{b_1^{(i)}+(T-1)(E^{(i)}-g_1^{(i)}),\bar{B}\}$.
The state evolves according to the function 
\[\big(b_1^{(i+1)},\, E^{(i+1)}\big)=\big(\min\{b_T^{(i)}-g_T^{(i)}+E^{(i+1)},\bar{B}\},\, E^{(i+1)}\big),\]
and the stage reward is given by
\begin{equation}
r(b_1^{(i)}, E^{(i)}, g_1^{(i)}, g_T^{(i)})
=\tfrac{T-1}{T}C(g_1^{(i)}) + \tfrac{1}{T}C(g_T^{(i)}).
\end{equation}
\label{lemma:MDP_reduction}
\end{lemma}

Additionally, it follows from the fact that $0\leq g_T^{(i)}\leq b_T^{(i)}$, $b_1^{(i+1)}=\min(b_T^{(i)}-g_T^{(i)}+E^{(i+1)},\bar{B})$ and $E^{(i+1)}$ is independent of the state (or equivalently by the principle of optimality~\cite{Bertsekas2001vol1}), that the policy for $g_T^{(i)}$ can be a function of $b_T^{(i)}$ instead of $(b_1^{(i)},E^{(i)})$.

The optimal policy can be found by solving the Bellman equation~\cite{Bertsekas2001vol2}, however this is hard to solve explicitly, even for the simple case of $T=1$ (i.i.d. energy arrivals).
Alternatively, it can be solved numerically using value iteration, but this can require extensive computation resources.
Specifically, since the state space is a continuous interval, and the action space is a two dimensional rectangle, only an approximate solution can be found. This is done by quantizing the state and actions spaces, a process which suffers from the curse of dimensionality.
Additionally, the numerical solution cannot provide insight into the structure of the optimal policy and the qualitative behavior of the optimal throughput, namely how it varies with the parameters of the problem.

In the next section, we propose an explicit online power control policy, and show that it is within a constant gap of $\frac{1}{2}\log e\approx 0.72$ to optimality, analogously to Theorem 2 of \cite{NearOptimal1} stated above.
This gap does not depend on any of the parameters of the problem, namely $\bar{B}$, $T$, or the distribution of the energy arrivals $P_E$.
Moreover, this policy yields a simple and insightful formula for the approximate throughput, 
which clarifies how the battery size needs to be chosen in terms of $T$ and $P_E$ for the resultant throughput to approach the AWGN capacity.

\section{Main Result}
\label{sec:main_result}

Note that if $E^{(i)} > \bar{B}$, then $b_1^{(i)}=b_2^{(i)}=\ldots=b_T^{(i)}=\bar{B}$ regardless of the allocated energy. Hence we can treat such energy arrivals as if $E^{(i)}=\bar{B}$, and for the rest of this section we will assume $E^{(i)}\leq\bar{B}$.

Before we formally state the main result of the paper, we informally motivate the policy we propose for the block i.i.d.\ energy arrival model. In the light of the discussion in the previous section, a natural way to extend the Fixed Fraction Policy of \cite{NearOptimal1} to the block i.i.d.\ model can be as follows: For an appropriately chosen $q\in[0,1]$, let
\begin{equation}\label{eq:ffp_simple}
g_j^{(i)}=\tfrac{q}{T}(b_1^{(i)}+(T-1)E^{(i)}),
\qquad j=1,\dots,T.
\end{equation}
The intuition behind this extension can be understood as follows: since the total energy to  be harvested throughout a block is known ahead of time in the first time slot of the block, this strategy decides on the total energy to be allocated in the current block $i$ by taking into account both the energy available in the battery in the first time slot of the block, $b_1^{(i)}$, and the energy that will be harvested in the remaining $T-1$ time slots, $(T-1)E^{(i)}$. The sum of these two quantities, i.e. the energy we already have in the battery plus the energy we known we will harvest, can be thought of as the energy we effectively have for this block. The total energy allocated to block $i$ is simply a fraction $q$ of the energy we effectively have. This total energy is then uniformly divided over the $T$ channel uses in the block, due to the concavity of the reward function akin to the optimal offline strategy.

We will adopt the policy in \eqref{eq:ffp_simple}, unless this allocation leads to an overflow of the battery during the block, and therefore a wasting of the harvested energy. Note that this strategy will not lead to a battery overflow throughout the block if and only if 
\begin{equation}\label{eq:battery_nooverflow_condition}
\big(1-(T-1)\tfrac{q}{T}\big)(b_1^{(i)}+(T-1)E^{(i)}) \leq \bar{B}.
\end{equation}
When this is the case, the battery state at the beginning of the last time slot of the block is given by
\begin{align}
b_T^{(i)}&=b_1^{(i)}+(T-1)(E^{(i)}-g_1^{(i)})\nonumber\\
&=\big(1-(T-1)\tfrac{q}{T}\big)(b_1^{(i)}+(T-1)E^{(i)}).
\end{align}
Therefore, when \eqref{eq:battery_nooverflow_condition} is satisfied, we can  write the policy in \eqref{eq:ffp_simple} in a way that agrees with the optimal policy structure for the MDP formulation discussed in Section~\ref{subsec:preliminary_results}:
\begin{equation}
\begin{aligned}
g_1^{(i)}&=\tfrac{q}{T}(b_1^{(i)}+(T-1)E^{(i)}),\\
g_T^{(i)}&=\tfrac{q}{q+(1-q)T}b_T^{(i)}.
\end{aligned}
\label{eq:naive_policy}
\end{equation}

For blocks in which the condition \eqref{eq:battery_nooverflow_condition} is not satisfied, we would want to modify the policy \eqref{eq:ffp_simple} so as not to waste the harvested energy. Note that this condition can be checked at the beginning of the block and the energy allocations can be increased from that in \eqref{eq:ffp_simple} if the condition is not satisfied. In particular, if \eqref{eq:battery_nooverflow_condition} is not satisfied, we modify the policy to:
\begin{equation}
\begin{aligned}
g_1^{(i)}&=\min\big(E^{(i)}-\tfrac{\bar{B}-b_1^{(i)}}{T-1}, \bar{B}\big),\\
g_T^{(i)}&=\tfrac{q}{q+(1-q)T}\bar{B}.
\end{aligned}
\label{eq:large_arrival_policy}
\end{equation}
The energy allocations in the first $T-1$ time slots are increased so that energy is not wasted, and the battery is fully charged after the last energy arrival, i.e. $b_T^{(i)}=\bar{B}$. 
Note that the energy allocated at the last time slot follows the same policy as \eqref{eq:naive_policy}, since $b_T^{(i)}=\bar{B}$.

While we can use the condition in \eqref{eq:battery_nooverflow_condition} for switching between the two modes of the policy, for small and large energy arrivals respectively, as discussed above, we would want to simplify this condition in a way that does not significantly degrade the performance but simplifies the following discussion.
If we assume the battery was empty at the end of the previous block, i.e. $b_1^{(i)}=E^{(i)}$, the  no battery overflow condition \eqref{eq:battery_nooverflow_condition} would be equivalent to $E^{(i)} \leq E_c$, where $E_c$ is a \emph{critical energy level} given by 
\begin{equation}
E_c=\frac{\bar{B}}{q+T(1-q)}.
\label{eq:def_Ec_alt}
\end{equation}
Note that when $E^{(i)} > E_c$, battery overflow will occur regardless of the state of the battery at the end of the previous block. Specifically, we propose to allocate energy according to \eqref{eq:large_arrival_policy} when $E^{(i)} > E_c$, and use \eqref{eq:naive_policy} when $E^{(i)} \leq E_c$.

It remains to choose the fixed fraction $q$. Recall that in the i.i.d.\ case, as discussed in Section~\ref{subsec:background}, $q$ was chosen to be $\frac{\mathbb{E}[\min(E_t,\bar{B})]}{\bar{B}}$, where $\bar{B}$ is the size of the battery
and $\mathbb{E}[\min(E_t,\bar{B})]$ is the average energy available in an epoch, which is the period between consecutive ``large arrival'' events $\{E_t>\bar{B}\}$.
In the block i.i.d.\ case, observe that when the event $\{E^{(i)}> E_c\}$ occurs at block~$i$, the battery will be fully charged at the end of the block.
Hence, let $p=\Pr(E^{(i)}> E_c)$, and imagine we put aside the first $T-1$ time slots of the ``large arrival'' block (in which we abandon the fixed fraction policy), and instead concentrate only on the subsequent slots (where we do apply it). The average energy available for this period can be computed as 
$\varepsilon=\bar{B}+(\frac{1}{p}-1)T\mathbb{E}\big[E^{(i)}\big|E^{(i)}\leq E_c\big]$,
because the battery is fully charged at the last time slot of the large arrival block, and at each one of the subsequent ``low-energy'' blocks the transmitter harvests an average amount of energy equal to $T\mathbb{E}\big[E^{(i)}\big|E^{(i)} \leq E_c\big]$. The duration of an epoch is, on average, $\tau=1 + T(\frac{1}{p}-1)$ slots. Therefore, the average energy available per time slot during this period is again given by $\varepsilon/\tau$. Note that the system is reset whenever an energy arrival larger than $E_c$ occurs, which leaves the battery fully charged at the beginning of the next epoch. Therefore, inspired by the i.i.d. case, given $E_c$ we may want to choose 
\begin{equation}
q=\frac{\varepsilon/\tau}{E_c}=\frac{\bar{B}+(\frac{1}{p}-1)T\mathbb{E}\big[E^{(i)}\big|E^{(i)} \leq E_c\big]}{E_c(1 + T(\frac{1}{p}-1))}.
\label{eq:def_q_alt}
\end{equation}
Recall however that given $q$, we want to choose $E_c$ as in \eqref{eq:def_Ec_alt}. These two desired relations for $E_c$ and $q$, along with the identity
\[
\mathbb{E}[\min(E^{(i)},E_c)]=p E_c+(1-p)\mathbb{E}[E^{(i)}|E^{(i)}\leq E_c],
\]
 can be solved to obtain the following equation: 
\begin{equation}
T E_c - (T-1)\mathbb{E}[\min(E^{(i)},E_c)] = \bar{B},
\label{eq:def_Ec}
\end{equation}
which can be solved for $E_c$ for given $T$, $\bar{B}$, and $P_E$ (it is shown in Appendix~\ref{sec:unique_Ec} that it has a unique solution in the interval $[0,\bar{B}]$).
Additionally, combining \eqref{eq:def_Ec_alt} and \eqref{eq:def_Ec} yields the following simple formula for $q$ given $E_c$: 
\begin{equation}
q = \frac{\mathbb{E}[\min(E^{(i)},E_c)]}{E_c}.
\label{eq:def_q}
\end{equation}
Note that this is essentially the same expression for $q$ as in the i.i.d.\ case, with $\bar{B}$ replaced by~$E_c$. Indeed, when $T=1$, eq. \eqref{eq:def_Ec} reduces to $E_c=\bar{B}$ and hence \eqref{eq:def_q} reduces to $\frac{\mathbb{E}[\min(E^{(i)},\bar{B})]}{\bar{B}}$.

To summarize, the online policy we propose for the block i.i.d case is given as follows.
\begin{policy}\label{policy}
 Given $T$, $\bar{B}$, and $P_E$ (the distribution of $E^{(i)}$), compute $E_c$ and $q$ according to \eqref{eq:def_Ec} and \eqref{eq:def_q}. Then apply
\begin{equation}
\begin{aligned}
g_1^{(i)}&=\begin{cases}
\tfrac{q}{T}(b_1^{(i)}+(T-1)E^{(i)}),& \text{ if } E^{(i)}\leq E_c,\\
E^{(i)}-\tfrac{\bar{B}-b_1^{(i)}}{T-1},& \text{ if } E_c< E^{(i)}\leq\bar{B},\\
\bar{B},& \text{ if }\bar{B}< E^{(i)},
\end{cases}\\
g_T^{(i)}&=\tfrac{q}{q+(1-q)T}b_T^{(i)},
\end{aligned}
\end{equation}
and note that $g_j^{(i)}=g_1^{(i)}$ for $j=2,\dots, T-1$.
\end{policy}
The main result of this paper is to prove that this policy is optimal within the same gap as in the i.i.d.\ case, as stated in the following theorem. 

\begin{theorem}\label{thm:throughput_bounds}
Let $E_c$ be the unique solution of~\eqref{eq:def_Ec}, and let $p=\Pr(E^{(i)}> E_c)$. Then the optimal throughput is bounded by 
\[\bar{\Theta}-\tfrac{1}{2}\log e\leq\Theta\leq\bar{\Theta},\]
where
\begin{equation}
\begin{aligned}
\bar{\Theta}
&=\tfrac{p(T-1)}{T}\mathbb{E}\left[\left.C\left(\min\left\{E^{(i)}-\tfrac{\bar{B}-E^{(i)}}{T-1},\,\bar{B}\right\}\right)\right|E^{(i)}> E_c\right]\\*
&\qquad +\tfrac{p+T(1-p)}{T}C\big(\mathbb{E}\big[\min(E^{(i)},E_c)\big]\big),
\end{aligned}
\label{eq:theta_bar_general}
\end{equation}
and the lower bound is achieved by Policy~\ref{policy}.
\end{theorem}

Note that the structure of the approximately optimal throughput expression has a natural interpretation in terms of Policy~\ref{policy}. 
The expression has two terms, corresponding to the two different operation modes of the policy.
The first term corresponds to the throughput achieved in the first $T-1$ time slots of a large energy arrival block. Note that these time slots correspond to a fraction $\frac{p(T-1)}{T}$ of the total time on average.
In the remaining fraction of the time, we apply the Fixed Fraction Policy, which, analogously to the i.i.d.\ case, achieves a throughput 
$C(\mathbb{E}[\min(E^{(i)},E_c)])$, where $\mathbb{E}[\min(E^{(i)},E_c)]=\varepsilon/\tau=q E_c$ is the average available energy rate for a low-energy period.

Theorem~\ref{thm:throughput_bounds} is proved by showing the throughput obtained by Policy~\ref{policy} under the process $E^{(i)}$ is lower bounded by the throughput obtained by Policy~\ref{policy} under a different block i.i.d. energy arrivals process $\hat{E}^{(i)}$, with a \emph{modified} distribution.
This modified distribution has structure similar to a Bernoulli distribution (the exact distribution will be precisely defined in Section~\ref{sec:lower_bound}).
The analysis of Policy~\ref{policy} turns out to be easier for Bernoulli distributions, since the regenerative structure discussed previously is inherent in the arrivals process.
Specifically, the case when $E^{(i)}\in\{0,\bar{E}\}$ for some $\bar{E}>0$ was solved in \cite{BlockiidBernoulli}, and was the basis for this work.

Denote $\mu=\mathbb{E}[E^{(i)}]$.
For all ergodic energy arrival processes (including block i.i.d.),
the AWGN capacity $\tfrac{1}{2}\log(1+\mu)$ is always an upper bound on the throughput, for any finite battery size.
However, as shown in the following corollary, in our block i.i.d.\ model this is nearly achievable only if the battery size is large enough.
\begin{corollary}
If $\bar{B}\geq \mu+T(E_{\max}-\mu)$, where $E^{(i)}\leq E_{\max}$ with probability 1, the approximate throughput reduces to
\begin{equation}
	\bar{\Theta}=C(\mu)=\tfrac{1}{2}\log(1+\mu).
	\label{eq:theta_bar_large_battery}
\end{equation}
\end{corollary}
\begin{proof}
Since $E^{(i)}\leq \mu+\frac{\bar{B}-\mu}{T}$, choose $E_c=\mu+\frac{\bar{B}-\mu}{T}$ and observe that $\mathbb{E}[\min(E^{(i)},E_c)]=\mu$ and therefore $E_c$ is the solution to~\eqref{eq:def_Ec}.
It follows that $p=0$ and \eqref{eq:theta_bar_general} reduces to \eqref{eq:theta_bar_large_battery}.
\end{proof}
We identify the case $\bar{B}\geq \mu+T(E_{\max}-\mu)$ as the \emph{large battery regime}.
The threshold $\mu+T(E_{\max}-\mu)$ can be intuitively interpreted as follows: When the battery size is infinite, it is straightforward to observe that the optimal policy is to allocate a constant amount of power equal to the mean energy arrival rate $g_t=\mu$ (cf. \cite{OzelUlukus2012,online_infiniteB4}).
Assume the battery was empty prior to the beginning of block $i$, i.e. $b_1^{(i)}=E^{(i)}$, and we apply this policy for all the time slots in block $i$. Then the battery level at the last time slot of the block will be $b_T^{(i)}=\mu+T(E^{(i)}-\mu)$.
This implies that we would need a battery size of at least $\mu+T(\bar{E}-\mu)$ in order to not waste energy due to an overflow, since $E^{(i)}$ can take values up to $E_{\max}$. 
However, the fact that we can nearly achieve the AWGN capacity as soon as the battery size is larger than this threshold is indeed surprising.


\subsection{Connection to Channel Capacity}

In this section, we will show how the approximate throughput of Theorem~\ref{thm:throughput_bounds} can provide an approximation to the information-theoretic capacity of the channel.
In this section we will use the notation $X_m^n=(X_m,X_{m+1},\ldots,X_n)$ for $m\leq n$, and $X^n=X_1^n$.

We consider an AWGN channel, i.e. the output is $Y_t=X_t+Z_t$, where $Z_t\sim\mathcal{N}(0,1)$ and $X_t\in\mathbb{R}$ is the input.
Instead of the energy constraints \eqref{eq:power} and \eqref{eq:battery}, which were only concerned with the amount of transmitted power at each time slot, the information-theoretic model imposes energy constraints on the amplitude of the transmitted symbol: at time $t$,
\begin{align}
X_t^2 &\leq B_t,\label{eq:inf_theory_amplitude}\\
B_t &= \min\{B_{t-1} - X_{t-1}^2 + E_t, \bar{B}\},\label{eq:inf_theory_battery}
\end{align}
where again we assume $B_1=\bar{B}$, and the harvesting process $\{E_t\}_{t=1}^{\infty}$ is the same as in Section~\ref{sec:system_model}.

Instead of a power control policy, our goal is to find a set of encoding functions and a decoding function:
\begin{align*}
&f_t^{\mathrm{enc}} : \mathcal{M} \times \mathcal{E}^t \to \mathcal{X},
\qquad t=1,\ldots,N,\\
&f^{\mathrm{dec}} : \mathcal{Y}^N \to \mathcal{M},
\end{align*}
where $\mathcal{M}=\{1,\ldots,2^{NR}\}$ is the message set and $\mathcal{X}=\mathcal{Y}=\mathbb{R}$ are the input and output spaces.
As usual, the capacity $C$ is the supremum of all rates $R$ for which the probability of error vanishes as $N\to\infty$ (see \cite{ShavivNguyenOzgur2016} for a detailed formulation of the problem).

The main result of this section is the following approximation to capacity.
\begin{theorem}\label{thm:capacity_bounds}
The capacity of the energy harvesting channel with block i.i.d.\ energy arrivals is bounded by
\begin{equation}
\bar{\Theta} - \tfrac{1}{T}H(E^{(i)}) - \tfrac{1}{2}\log\big(\tfrac{\pi e^2}{2}\big) \leq C \leq \bar{\Theta},
\end{equation}
where $\bar{\Theta}$ is given by \eqref{eq:theta_bar_general}.
\end{theorem}
This is the counterpart of Theorem~1 in \cite{ShavivNguyenOzgur2016}, in which the i.i.d.\ case ($T=1$) is considered.
The proof follows similarly, with some modifications to account for a block i.i.d.\ energy arrivals process.
Specifically, we have the following intermediate result.
\begin{proposition}\label{prop:capacity_throughput_connection}
The capacity of the energy harvesting channel with block i.i.d.\ energy arrivals is bounded by
\begin{equation}
\Theta - \tfrac{1}{T}H(E^{(i)}) - \tfrac{1}{2}\log\big(\tfrac{\pi e}{2}\big) \leq C \leq \Theta,
\label{eq:capacity_throughput_connection}
\end{equation}
where $\Theta$ is the optimal throughput defined in \eqref{eq:def_optimal_throughput}.
\end{proposition}
See Appendix~\ref{sec:inf_theory_gap} for the proof.
Clearly, applying Theorem~\ref{thm:throughput_bounds} to \eqref{eq:capacity_throughput_connection} yields Theorem~\ref{thm:capacity_bounds}.
\begin{remark}
The lower bound in \eqref{eq:capacity_throughput_connection} can be tightened as in \cite[Thm.~4]{ShavivNguyenOzgur2016} to 
\[
C \geq \mathscr{T}(\mathbf{g}) - H(\mathbf{g}) - \tfrac{1}{2}\log\big(\tfrac{\pi e}{2}\big),
\]
for any policy $\mathbf{g}$,
where $H(\mathbf{g})=\limsup_{N\to\infty}\tfrac{1}{N}H(g^N(E^N))$ is the entropy rate of the process $\{g_t(E^t)\}_{t=1}^{\infty}$.
This allows choosing policies for which $H(\mathbf{g})$ is bounded by a constant which is independent of the statistics of $E^{(i)}$, thereby making the gap dependent on the entropy rate of the power control process instead of the entropy rate of the energy harvesting process itself. It is shown in~\cite{ShavivNguyenOzgur2016} that, in the i.i.d.\ ($T=1$) case, one can design online power control policies that have constant entropy rate (in particular, $H(\mathbf{g})\leq 1 $ bits per channel use, independent of the distribution of the energy arrivals) and at the same time are within a constant gap to the optimal throughput. This suggests that in the i.i.d. case the information-theoretic capacity of the channel can be approximated by the optimal throughput within a constant gap independent of the parameters of the problem. In the block i.i.d. case, the entropy rate of Policy~\ref{policy} we develop in the previous section can not be simply bounded by a constant independent of the distribution of the energy arrivals. We believe it is possible to modify Policy~\ref{policy} in the lines of \cite{ShavivNguyenOzgur2016} to obtain a policy with a constant entropy rate, and hence show that the information-theoretic capacity and the optimal online throughput are indeed within a constant gap of each other, independent of the parameters of the problem. However, this is less critical for the block i.i.d. case since the entropy rate of the energy arrival process decreases to zero with increasing blocklength $T$. Hence,  we expect the gap term $\tfrac{1}{T}H(E^{(i)})$ in \eqref{eq:capacity_throughput_connection} to be small for sufficiently large values of $T$.
\end{remark}

In the remainder of the paper, we provide the proof of Theorem~\ref{thm:throughput_bounds}.
In Section~\ref{sec:lower_bound} we prove the lower bound on the throughput, and in Section~\ref{sec:upper_bound} we derive an upper bound which differs from the lower bound by no more than 0.72 bits per channel use. 
Section~\ref{sec:conclusion} concludes the paper.

\section{Lower Bound}
\label{sec:lower_bound}

In this section we will show that Policy~\ref{policy} achieves a throughput which is lower bounded by $\bar{\Theta}-\frac{1}{2}\log e$ as defined in Theorem~\ref{thm:throughput_bounds}.
This will be done in three parts. In the first part, we will derive the lower bound for the special case of Bernoulli energy arrivals, i.e. $E^{(i)}\in\{0,\bar{E}\}$ for some $\bar{E}>0$.
In the second part, we will generalize this to a larger class of energy arrival distributions, dubbed \emph{semi-Bernoulli} distributions, which maintain the regenerative structure exhibited by Bernoulli energy arrivals.
Finally, in the third part we will show that for any block i.i.d.\ energy arrivals distribution, we can find a modified distribution of the form discussed in the second part, for which the throughput under Policy~\ref{policy} is a lower bound to the throughput under the original distribution. 

\subsection{Bernoulli Energy Arrivals}
\label{subsec:lower_bound_bernoulli}

We start with the special case of Bernoulli energy arrivals, namely
\begin{equation}
E^{(i)} = \begin{cases}
\bar{E} &, \text{w.p. } p_0,\\
0 &, \text{w.p. } 1-p_0,
\end{cases}
\label{eq:bernoulli_distribution}
\end{equation}
for some $0\leq p_0\leq 1$.

We will assume the energy level $\bar{E}$ satisfies 
$E_c\leq\bar{E} \leq \bar{B}$,
 where $E_c$ is the solution to \eqref{eq:def_Ec}.
This implies $\mathbb{E}[\min(E^{(i)}, E_c)]=p_0E_c$, which yields $E_c=\frac{\bar{B}}{p_0+T(1-p_0)}$ and $q=p_0$ from \eqref{eq:def_Ec} and \eqref{eq:def_q}, respectively.
In the following proposition, we show that the lower bound in Theorem~\ref{thm:throughput_bounds} holds for this special case.
\begin{proposition}\label{prop:bernoulli_lower_bound}
Let $E^{(i)}$ be distributed Bernoulli as in \eqref{eq:bernoulli_distribution} with $E_c \leq \bar{E}\leq\bar{B}$. Then the throughput obtained by Policy~\ref{policy} is lower bounded by
\[
\mathscr{T}(\mathbf{g}) \geq \bar{\Theta} - \tfrac{1}{2}\log e.
\]
\end{proposition}
\begin{proof}
Assume for now $\bar{E} > E_c$; the special case $\bar{E} = E_c$ will be treated afterwards.
We have $p=\Pr(E^{(i)}>E_c)=p_0$, and the approximate throughput in \eqref{eq:theta_bar_general} is given by
\begin{equation}
\bar{\Theta} = \tfrac{p_0(T-1)}{T}C(\bar{E}-\tfrac{\bar{B}-\bar{E}}{T-1}) + \tfrac{p_0+T(1-p_0)}{T}C(p_0 E_c).
\label{eq:theta_bar_bernoulli_large}
\end{equation}

It can be observed that Policy~\ref{policy} reduces to
\begin{align*}
g_1^{(i)} &= 
\begin{cases}
	\tfrac{p_0}{T}b_1^{(i)}, &\text{ if }E^{(i)}=0,\\
	\bar{E}-\tfrac{\bar{B}-b_1^{(i)}}{T-1}, &\text{ if }E^{(i)}=\bar{E},
\end{cases}\\
g_T^{(i)} &= \tfrac{p_0}{p_0+(1-p_0)T}b_T^{(i)}.
\end{align*}
We see that whenever there is a positive energy arrival, the battery will be fully charged to $\bar{B}$ by the end of the block.

Consider the Markov reward process \cite[Ch.~8.2]{Puterman2005} obtained by applying Policy~\ref{policy}. This comprises of the state process $(b_1^{(i)}, E^{(i)})$ and a reward function given by $r(b_1^{(i)},E^{(i)})=\tfrac{T-1}{T}C(g_1^{(i)})+\tfrac{1}{T}C(g_T^{(i)})$.
It can be verified that the reward function is given by
\[
r(b_1^{(i)},E^{(i)}) = 
\begin{cases}
	C(\tfrac{p_0}{T}b_1^{(i)}), 
	&\text{ if }E^{(i)}=0,\\
	\lefteqn{\tfrac{T-1}{T}C(\bar{E}-\tfrac{\bar{B}-b_1^{(i)}}{T-1})+\tfrac{1}{T}C(p_0E_c),}\\
	\hspace{110pt} &\text{ if }E^{(i)}=\bar{E}.
\end{cases}
\]
The battery state evolves according to
\begin{align}
b_1^{(i+1)} &= \min(b_T^{(i)}-g_T^{(i)}+E^{(i+1)}, \bar{B}),\nonumber\\
b_T^{(i)}-g_T^{(i)} &= 
\begin{cases}
	(1-p_0)b_1^{(i)}, &\text{ if }E^{(i)}=0,\\
	T(1-p_0)E_c, &\text{ if }E^{(i)}=\bar{E}.
\end{cases}
\label{eq:bernoulli_battery_evolution}
\end{align}

In what follows, we will lower bound the throughput obtained by Policy~\ref{policy} by analyzing the average long-term reward under the following reward process:
\[
\tilde{r}(b_1^{(i)},E^{(i)}) = 
\begin{cases}
	C(\tfrac{p_0}{T}b_1^{(i)}), &\text{if }E^{(i)}=0,\\
	\tfrac{T-1}{T}C(\bar{E}-\tfrac{\bar{B}-\bar{E}}{T-1})+\tfrac{1}{T}C(p_0 E_c), &\text{if }E^{(i)}=\bar{E}.
\end{cases}
\]
Note that $r(b_1^{(i)},E^{(i)})\geq \tilde{r}(b_1^{(i)},E^{(i)})$ since $b_1^{(i)}\geq E^{(i)}$.
Hence the throughput obtained by Policy~\ref{policy} is lower bounded by
\begin{align}
\mathscr{T}(\mathbf{g})
&= \liminf_{N\to\infty}\tfrac{1}{N}\sum_{i=1}^{N}r(b_1^{(i)},E^{(i)})\nonumber\\
&\geq \liminf_{N\to\infty}\tfrac{1}{N}\sum_{i=1}^{N}\tilde{r}(b_1^{(i)},E^{(i)}).\label{eq:markov_reward_process_lower_bound}
\end{align}

Observe that the process $\tilde{r}(b_1^{(i)},E^{(i)})$ is a \emph{regenerative} process, where regeneration occurs whenever $E^{(i)}=\bar{E}$. 
Then, by the renewal reward theorem
(see e.g. Theorem 3.1 in~\cite[Ch.~VI]{asmussen2008applied}, \cite[Prop.~7.3]{ross2014introduction}, or~\cite[Lemma~1]{NearOptimal1}), eq.~\eqref{eq:markov_reward_process_lower_bound} becomes:
\begin{align}
\mathscr{T{}}(\mathbf{g})
&\geq\frac{1}{\mathbb{E}L}\mathbb{E}\left[\sum_{i=1}^{L}
\tilde{r}({b}_1^{(i)},E^{(i)})\right],
\label{eq:bernoulli_throughput_first_lower_bound}
\end{align}
where $E^{(1)}=\bar{E}$, $E^{(i)}=0$ for $i>1$, and $L$ is a Geometric RV with parameter $p_0$, representing the number of blocks between consecutive positive energy arrivals.
From \eqref{eq:bernoulli_battery_evolution}, it follows that $b_1^{(i)}=(1-p_0)^{i-1}TE_c$ for $i\geq 2$, and therefore \eqref{eq:bernoulli_throughput_first_lower_bound} becomes
\begin{equation}
\begin{aligned}
\mathscr{T}(\mathbf{g})
&\geq \frac{1}{\mathbb{E}L}\mathbb{E}\Big[
\tfrac{T-1}{T}C(\bar{E}-\tfrac{\bar{B}-\bar{E}}{T-1})+\tfrac{1}{T}C(p_0 E_c)\\*
&\hspace{70pt}
+\sum_{i=2}^{L}C(p_0(1-p_0)^{i-1}E_c)
\Big].
\end{aligned}
\label{eq:bernoulli_throughput_second_lower_bound}
\end{equation}
Using the fact that 
$\log(1+\alpha x) \geq \log(1+x) + \log\alpha$
for $0\leq \alpha\leq1$, we can lower bound
\[
C(p_0(1-p_0)^{i-1}E_c) \geq C(p_0 E_c) + \tfrac{i-1}{2}\log(1-p_0).
\]
Substituting in \eqref{eq:bernoulli_throughput_second_lower_bound}:
\begin{align*}
\mathscr{T}(\mathbf{g})
&\geq \frac{1}{\mathbb{E}L}\mathbb{E}\Big[
\tfrac{T-1}{T}C(\bar{E}-\tfrac{\bar{B}-\bar{E}}{T-1})+\tfrac{1}{T}C(p_0 E_c)\\*
&\hspace{50pt}
+(L-1)C(p_0 E_c)+\tfrac{L^2-L}{4}\log(1-p_0)
\Big]\\
&= \tfrac{p_0(T-1)}{T}C(\bar{E}-\tfrac{\bar{B}-\bar{E}}{T-1})
+(\tfrac{p_0}{T}+1-p_0)C(p_0E_c)\\*
&\qquad
+\tfrac{1-p_0}{2p_0}\log(1-p_0)\\
&\geq \bar{\Theta} - \tfrac{1}{2}\log e,
\end{align*}
where the last step is due to \eqref{eq:theta_bar_bernoulli_large} and because $\tfrac{1-p}{2p}\log(1-p) \geq -\tfrac{1}{2}\log e$ for $0\leq p\leq 1$.

It remains to show the lower bound holds for $\bar{E} = E_c$. In this case, we have $p = \Pr(E^{(i)} > E_c) = 0$, and the approximate throughput \eqref{eq:theta_bar_general} is
\begin{equation}
	\bar{\Theta} = C(p_0 E_c).
\end{equation}
Policy~\ref{policy} takes the form
\begin{align*}
g_1^{(i)} &= 
\begin{cases}
	\tfrac{p_0}{T}b_1^{(i)}, &\text{ if }E^{(i)}=0,\\
	\tfrac{p_0}{T}(b_1^{(i)}+(T-1)E_c), &\text{ if }E^{(i)}=E_c,
\end{cases}\\
g_T^{(i)} &= \tfrac{p_0}{p_0+(1-p_0)T}b_T^{(i)}.
\end{align*}
Observe that here as well, whenever there is a positive energy arrival the battery will be fully charged to $\bar{B}$ by the end of the block. This can be seen by computing the battery state at the end of the block when $E^{(i)}=E_c$:
\begin{align*}
b_T^{(i)} &= \min\left\{(1-\tfrac{p_0(T-1)}{T})(b_1^{(i)}+(T-1)E_c),\ \bar{B}\right\}.
\end{align*}
Since $b_1^{(i)} \geq E^{(i)} = E_c$:
\begin{align*}
(1-\tfrac{p_0(T-1)}{T})(b_1^{(i)}+(T-1)E_c)
&\geq (1-\tfrac{p_0(T-1)}{T})TE_c\\
&= \bar{B}.
\end{align*}
Therefore the battery state evolves exactly the same as in the case $\bar{E} > E_c$.

The reward function is given by
\[
r(b_1^{(i)}, E^{(i)}) =
\begin{cases}
	C(\tfrac{p_0}{T}b_1^{(i)}), &\text{ if }E^{(i)}=0,\\
	\lefteqn{\tfrac{T-1}{T}C(\tfrac{p_0}{T}(b_1^{(i)}+(T-1)E_c)) + \tfrac{1}{T}C(p_0 E_c),}\\
	\hspace{115pt} &\text{ if }E^{(i)}=E_c.
\end{cases}
\]
As before, we will lower bound the reward function using the fact that $b_1^{(i)} \geq E^{(i)}$, giving
\[
\tilde{r}(b_1^{(i)} , E^{(i)}) =
\begin{cases}
	C(\tfrac{p_0}{T}b_1^{(i)}), &\text{ if }E^{(i)}=0,\\
	C(p_0 E_c), &\text{ if }E^{(i)}=E_c.\\
\end{cases}
\]
Repeating the previous steps, we obtain
\begin{align*}
\mathscr{T}(\mathbf{g})
&\geq C(p_0E_c)
+\tfrac{1-p_0}{2p_0}\log(1-p_0)\\
&\geq \bar{\Theta} - \tfrac{1}{2}\log e.
\end{align*}
\end{proof}

\subsection{Semi-Bernoulli Energy Arrivals}
\label{subsec:lower_bound_semi_bernoulli}

\begin{figure*}[!t]
\normalsize
\setcounter{mytempeqncnt}{\value{equation}}
\setcounter{equation}{33}
\begin{align}
\bar{\Theta} &= \tfrac{p(T-1)}{T}\mathbb{E}\left[\left.C\left(\min\left\{E^{(i)}-\tfrac{\bar{B}-E^{(i)}}{T-1},\bar{B}\right\}\right)\right|E^{(i)}>E_c\right]
	+ \tfrac{p+T(1-p)}{T}C(q E_c)\label{eq:theta_bar_semi_bernoulli_alt}\\
&= \tfrac{q(T-1)}{T}\mathbb{E}\left[\left.C\left(\min\left\{E^{(i)}-\tfrac{\bar{B}-E^{(i)}}{T-1},\bar{B}\right\}\right)\right|E^{(i)}\geq E_c\right]
	 -\tfrac{T-1}{T}\Pr(E^{(i)}=E_c)\cdot C(E_c-\tfrac{\bar{B}-E_c}{T-1})
	+\tfrac{p+T(1-p)}{T}C(q E_c)\\
&= \tfrac{q(T-1)}{T}\mathbb{E}\left[\left.C\left(\min\left\{E^{(i)}-\tfrac{\bar{B}-E^{(i)}}{T-1},\bar{B}\right\}\right)\right|E^{(i)}\geq E_c\right]
	+\tfrac{T-(T-1)p-(T-1)\Pr(E^{(i)}=E_c)}{T} C(q E_c)
	\label{eq:theta_bar_semi_bernoulli_step}\\
&= \tfrac{q(T-1)}{T}\mathbb{E}\left[\left.C\left(\min\left\{E^{(i)}-\tfrac{\bar{B}-E^{(i)}}{T-1},\bar{B}\right\}\right)\right|E^{(i)}\geq E_c\right]
	+\tfrac{q+T(1-q)}{T} C(q E_c),
	\label{eq:theta_bar_semi_bernoulli}
\end{align}
\setcounter{equation}{\value{mytempeqncnt}}
\hrulefill
\vspace*{4pt}
\end{figure*}

Now consider an energy arrivals process $E^{(i)}$, and let $E_c$ be the solution to \eqref{eq:def_Ec}.
We say that $E^{(i)}$ has a \emph{semi-Bernoulli} distribution if 
\begin{equation}
\Pr(0 < E^{(i)} < E_c) = 0.
\end{equation}
In other words, $E^{(i)}$ can either be 0 or take a value which is at least $E_c$, but not any value in the open interval $(0, E_c)$.
Note that the Bernoulli distribution from the previous section, namely $E^{(i)}\in\{0,\bar{E}\}$ with $E_c\leq\bar{E}$, clearly satisfies this condition; however, in general, $E^{(i)}$ can take an arbitrary number of values.

Observe that $\mathbb{E}[\min(E^{(i)}, E_c)]=E_c\cdot\Pr(E^{(i)}\geq E_c)$, which from \eqref{eq:def_q} yields $q=\Pr(E^{(i)}\geq E_c)$.
Denote $p=\Pr(E^{(i)}>E_c)$ as in Theorem~\ref{thm:throughput_bounds}.
Note that $q=p+\Pr(E^{(i)}=E_c)$.
The approximate throughput $\bar{\Theta}$ in \eqref{eq:theta_bar_general} is given by \eqref{eq:theta_bar_semi_bernoulli_alt}--\eqref{eq:theta_bar_semi_bernoulli} at the top of the page,
where \eqref{eq:theta_bar_semi_bernoulli_step} is due to \eqref{eq:def_Ec_alt}.
\addtocounter{equation}{4}
The following proposition generalizes the result of Proposition~\ref{prop:bernoulli_lower_bound} to semi-Bernoulli distributions.
\begin{proposition}\label{prop:semi_bernoulli_lower_bound}
Suppose $E^{(i)}$ and $E_c$ satisfy $\Pr(0 < E^{(i)} < E_c) = 0$.
Then the throughput obtained by Policy~\ref{policy} is lower bounded by
\[
\mathscr{T}(\mathbf{g}) \geq \bar{\Theta} - \tfrac{1}{2}\log e.
\]
\end{proposition}
\begin{proof}
Observe that Policy~\ref{policy} reduces to
\begin{align*}
g_1^{(i)} &= 
\begin{cases}
	\tfrac{q }{T}b_1^{(i)}, &\text{ if }E^{(i)} = 0,\\
	\tfrac{q }{T}(b_1^{(i)}+(T-1)E_c), &\text{ if }E^{(i)}=E_c,\\
	E^{(i)} - \tfrac{\bar{B} - b_1^{(i)}}{T-1}, &\text{ if }E_c < E^{(i)} \leq \bar{B},	\\
	\bar{B}, &\text{ if }\bar{B} < E^{(i)},
\end{cases}\\
g_T^{(i)} &= \tfrac{q }{q +(1-q )T}b_T^{(i)}.
\end{align*}
As in the previous section, whenever there is a positive energy arrival (or equivalently $E^{(i)} \geq E_c$), the battery will be fully charged to $\bar{B}$ by the end of the block.
Taking the appropriate Markov reward process with reward function
\[
r(b_1^{(i)}, E^{(i)}) = 
\begin{cases}
	C(\tfrac{q }{T}b_1^{(i)}), 
		&\text{ if }E^{(i)}=0,\\
	\lefteqn{\tfrac{T-1}{T}C(\tfrac{q }{T}(b_1^{(i)}+(T-1)E_c))+\tfrac{1}{T}C(q E_c),}\\
	 	&\text{ if }E^{(i)}=E_c,\\
	\lefteqn{\tfrac{T-1}{T}C(E^{(i)}-\tfrac{\bar{B}-b_1^{(i)}}{T-1})+\tfrac{1}{T}C(q  E_c),}\\
	 	&\text{ if }E_c < E^{(i)} \leq \bar{B},\\
	\tfrac{T-1}{T}C(\bar{B}) + \tfrac{1}{T}C(q  E_c), 
		\hspace{10pt}
		&\text{ if }\bar{B} < E^{(i)},
\end{cases}
\]
we can again obtain a lower bound on the reward function by using the relation $b_1^{(i)} \geq E^{(i)}$ for the appropriate cases (specifically when $E^{(i)}>0$):
\[
\tilde{r}(b_1^{(i)},E^{(i)}) =
\begin{cases}
	C(\tfrac{q }{T}b_1^{(i)}), 
		&\text{ if }E^{(i)}=0,\\
	C(q E_c), 
		&\text{ if }E^{(i)}=E_c,\\
	\lefteqn{\tfrac{T-1}{T}C(E^{(i)}-\tfrac{\bar{B}-E^{(i)}}{T-1})+\tfrac{1}{T}C(q E_c),} \\
		&\text{ if }E_c<E^{(i)}\leq\bar{B},\\
	\tfrac{T-1}{T}C(\bar{B})+\tfrac{1}{T}C(q E_c), 
		\hspace{10pt}
		&\text{ if }E^{(i)} > \bar{B}.
\end{cases}
\]
This can be written succinctly as follows:
\[
\tilde{r}(b_1^{(i)} ,E^{(i)}) =
\begin{cases}
	C(\tfrac{q }{T}b_1^{(i)}), 
		&\text{ if }E^{(i)}=0,\\
	\lefteqn{\tfrac{T-1}{T}C(\min\{E^{(i)}-\tfrac{\bar{B}-E^{(i)}}{T-1},\bar{B}\})+\tfrac{1}{T}C(q  E_c),}\\
		\hspace{110pt}
	 	&\text{ if }E^{(i)} \geq E_c.
\end{cases}
\]

The process $\tilde{r}(b_1^{(i)}, E^{(i)})$ is regenerative, with regeneration occurring at the event $E^{(i)} \geq E_c$.
The renewal reward theorem takes the form
\begin{align*}
\lefteqn{\liminf_{N\to\infty}\frac{1}{N}\sum_{i=1}^{N}\tilde{r}(b_1^{(i)},E^{(i)})}\\
& = 
\frac{1}{\mathbb{E}L}\mathbb{E}\Big[\sum_{i=1}^{L}\tilde{r}(b_1^{(i)},E^{(i)})\Big| E^{(1)} \geq E_c,\ E^{(i)} = 0,\ i\geq 2 \Big].
\end{align*}
The derivation carries on almost identically to the proof of Proposition~\ref{prop:bernoulli_lower_bound}, where the only difference is in the first term of \eqref{eq:bernoulli_throughput_second_lower_bound}.
We obtain \eqref{eq:semi_bernoulli_throughput_first}--\eqref{eq:semi_bernoulli_throughput_final} at the top of the page,
where \eqref{eq:semi_bernoulli_throughput_final} is due to \eqref{eq:theta_bar_semi_bernoulli}.
\addtocounter{equation}{3}
\end{proof}

\subsection{General Block i.i.d.\ Energy Arrivals}
\label{subsec:lower_bound_general}

\begin{figure*}[!t]
\normalsize
\setcounter{mytempeqncnt}{\value{equation}}
\setcounter{equation}{37}
\begin{align}
\mathscr{T}(\mathbf{g})
&\geq \frac{1}{\mathbb{E}L}\mathbb{E}\left[\left.\tfrac{T-1}{T}C\left(\min\left\{E^{(1)}-\tfrac{\bar{B}-E^{(1)}}{T-1},\bar{B}\right\}\right)  +\tfrac{1}{T}C(q  E_c)\right|E^{(1)}\geq E_c\right] 
	+\frac{1}{\mathbb{E}L}\mathbb{E}\left[\sum_{i=2}^{L}C(q (1-q )^{i-1}E_c)\right]
	\label{eq:semi_bernoulli_throughput_first}\\
&\geq \tfrac{q (T-1)}{T}\mathbb{E}\left[\left.C\left(\min\left\{E^{(i)}-\tfrac{\bar{B}-E^{(i)}}{T-1},\bar{B}\right\}\right)\right|E^{(i)}\geq E_c\right] + \tfrac{q +T(1-q )}{T}C(q  E_c) - \tfrac{1}{2}\log e\\
&= \bar{\Theta} - \tfrac{1}{2}\log e,
	\label{eq:semi_bernoulli_throughput_final}
\end{align}
\setcounter{equation}{\value{mytempeqncnt}}
\hrulefill
\vspace*{4pt}
\end{figure*}

Now consider an arbitrary distribution of energy arrivals $E^{(i)}$.
We lower bound the throughput obtained by Policy~\ref{policy} using a technique inspired by \cite{NearOptimal1}. There, the throughput obtained by the proposed policy was lower bounded by showing that the Bernoulli harvesting process yields the \emph{worst} performance compared to all other i.i.d.\ processes with the same mean.
Accordingly, we suggest a mean-preserving modification to the energy arrival distribution; specifically, the modified distribution will be semi-Bernoulli, as defined in the previous section.
We then show that the throughput obtained by Policy~\ref{policy} under this modified distribution is lower than under the original distribution.
Subsequently, the throughput obtained under the modified harvesting process is readily lower bounded using Proposition~\ref{prop:semi_bernoulli_lower_bound}.

We begin by defining the \emph{modified} energy arrival process:
\begin{equation}
\hat{E}^{(i)} := W\cdot 1\{E^{(i)}\leq E_c\} + E^{(i)}\cdot 1\{E^{(i)}>E_c\},
\label{eq:def_Ehat}
\end{equation}
where $1\{\,\cdot\,\}$ is the indicator function, and
$W$ is a Bernoulli RV independent of $E^{(i)}$, with $W\in\{0,E_c\}$.
To make sure the mean is preserved, i.e. $\mathbb{E}[\hat{E}^{(i)}]=\mathbb{E}[E^{(i)}]$, we write
\begin{align*}
\mathbb{E}[E^{(i)}\cdot 1\{E^{(i)}\leq E_c\}]
&= \mathbb{E}[W\cdot 1\{E^{(i)}\leq E_c\}]\\
&= \Pr(W=E_c)\cdot E_c\cdot \Pr(E^{(i)}\leq E_c),
\end{align*}
giving
\begin{equation}
\Pr(W=E_c) = \frac{\mathbb{E}[E^{(i)} | E^{(i)} \leq E_c]}{E_c}.
\label{eq:def_W}
\end{equation}
With this distribution of $W$, the probability of positive energy arrival $\hat{E}^{(i)}$ is given by:
\begin{align}
\Pr(\hat{E}^{(i)}>0)
&= \Pr(\hat{E}^{(i)} \geq E_c)\nonumber\\
&= (1-p) \Pr(W=E_c) + p\nonumber\\
&= \frac{\mathbb{E}[E^{(i)}\cdot 1\{E^{(i)}\leq E_c\}]}{E_c} + \mathbb{E}[1\{E^{(i)}>E_c\}]\nonumber\\
&= \frac{\mathbb{E}[\min(E^{(i)},E_c)]}{E_c}\nonumber\\
&= q.\label{eq:prob_Ehat_positive}
\end{align}

In what follows, we will analyze the long-term average throughput obtained by Policy~\ref{policy} under the original harvesting process $E^{(i)}$, as well as under the modified process $\hat{E}^{(i)}$.
As before, we consider the Markov reward process induced by applying Policy~\ref{policy}.
It would be convenient to describe the system using simpler state variables, defined below:
\begin{align}
x_i&=\frac{b_1^{(i)}+(T-1)E^{(i)}}{T},\\
s_i&=1\big\{E^{(i)}>E_c\big\}.
\end{align}
Using these new state variables, Policy~\ref{policy} can be expressed as follows:
\begin{align*}
g_1^{(i)} &= \begin{cases}
	q x_i, &\text{ if }s_i=0,\\
	\min\big\{\tfrac{Tx_i-\bar{B}}{T-1}, \bar{B}\big\}, &\text{ if }s_i=1,
\end{cases}\\
g_T^{(i)} &= \begin{cases}
	q \cdot \min(x_i, E_c), &\text{ if }s_i=0,\\
	q E_c, &\text{ if }s_i=1.
\end{cases}
\end{align*}
Accordingly, the reward function is given by
\begin{equation}
r(x_i,s_i)=\begin{cases}
\tfrac{T-1}{T}C(qx_i)+\tfrac{1}{T}C(q\cdot\min(x_i,E_c)), &\text{ if }s_i=0,\\
\frac{T-1}{T}C\left(\min\left\{\frac{Tx_i-\bar{B}}{T-1},\bar{B}\right\}\right)
+\frac{1}{T}C\left(qE_c\right)
&,\text{if }s_i=1,
\end{cases}
\label{eq:reward}
\end{equation}
and the state dynamics are given by\footnote{While the state variables were changed from $(b_1^{(i)}, E^{(i)})$ to $(x_i,s_i)$, the disturbance is still $E^{(i+1)}$, which is independent of the current state.}
\begin{align}
x_{i+1}&=\begin{cases}
\lefteqn{\min\big\{(1-q)\min(x_i,E_c)+E^{(i+1)},\,\tfrac{\bar{B}+(T-1)E^{(i+1)}}{T}\big\},}\\
&\text{if }s_i=0,\\
\lefteqn{\min\big\{(1-q)E_c+E^{(i+1)},\,\tfrac{\bar{B}+(T-1)E^{(i+1)}}{T}\big\},}\\
\hspace{160pt}&\text{if }s_i=1,
\end{cases}\label{eq:x_state_dynamics}\\
s_{i+1}&=1\big\{E^{(i+1)}>E_c\big\}.\label{eq:s_state_dynamics}
\end{align}

Define the $N$-horizon total throughput obtained by Policy~\ref{policy}, when the initial state is $(x,s)$:
\begin{equation}
J_N(x,s) = \mathbb{E}\left[\left.\sum_{i=1}^{N}r(x_i,s_i)\right|x_1=x, s_1=s\right].
\end{equation}
The long-term average expected throughput obtained by Policy~\ref{policy} is given by 
\[
\mathscr{T}(\mathbf{g}) = \liminf_{N\to\infty}\frac{1}{N}J_N(x,s),
\]
for any $x\in[0,\bar{B}]$, $s\in\{0,1\}$.

For the modified harvesting process, we similarly define the processes $\hat{x}_i$ and $\hat{s}_i$ given by~\eqref{eq:x_state_dynamics} and \eqref{eq:s_state_dynamics}, respectively, by replacing $E^{(i+1)}$ with $\hat{E}^{(i+1)}$.
The $N$-horizon total throughput obtained by Policy~\ref{policy} under the modified energy arrival process is given by
\[
\hat{J}_N(x,s)=\mathbb{E}\left[\left.\sum_{i=1}^{N}r(\hat{x}_i,\hat{s}_i)\right|\hat{x}_1=x,\hat{s}_1=s\right],
\]
and the long-term average throughput (for which we provided a lower bound in Proposition~\ref{prop:semi_bernoulli_lower_bound} in the previous section), is given by $\hat{\mathscr{T}}(\mathbf{g})=\liminf_{N\to\infty}\frac{1}{N}\hat{J}_N(x,s)$.

In the following proposition, we claim that the $N$-horizon expected throughput for the original distribution of block i.i.d.\ energy arrivals $E^{(i)}$ is greater than the throughput obtained for the modified distribution $\hat{E}^{(i)}$, for any $N$ and any initial state $(x,s)$.
\begin{proposition}\label{prop:semi_bernoulli_is_worst}
For any $x\in[0,\bar{B}]$, $s\in\{0,1\}$, and integer $N\geq1$:
\begin{equation}
J_N(x,s)\geq\hat{J}_N(x,s).
\end{equation}
\end{proposition}
This is proved by induction, making use of the concavity and monotonicity of the reward function $r(x,s)$. We defer the proof to Appendix~\ref{sec:semi_bernoulli_is_worst}.

By taking $N\to\infty$, an immediate corollary of Proposition~\ref{prop:semi_bernoulli_is_worst} is 
\[
\mathscr{T}(\mathbf{g}) \geq \hat{\mathscr{T}}(\mathbf{g}).
\]
Since $\hat{E}^{(i)}$ is a semi-Bernoulli process, i.e. it satisfies $\Pr(0<\hat{E}^{(i)}<E_c)=0$, we can readily obtain a lower bound on $\hat{\mathscr{T}}(\mathbf{g})$ by applying Proposition~\ref{prop:semi_bernoulli_lower_bound} from the previous section:
\begin{align*}
\hat{\mathscr{T}}(\mathbf{g})
&\geq \tfrac{p(T-1)}{T}\mathbb{E}\Big[C\big(\min\big\{\hat{E}^{(i)}-\tfrac{\bar{B}-\hat{E}^{(i)}}{T-1},\bar{B}\big\}\big)\Big|\hat{E}^{(i)}>E_c\Big]\\*
&\qquad + \tfrac{p+T(1-p)}{T}C(E_c\cdot\Pr(\hat{E}^{(i)}\geq E_c)) - \tfrac{1}{2}\log e,
\end{align*}
where the above expression for the approximate throughput $\bar{\Theta}$ of the modified process $\hat{E}^{(i)}$ is given by \eqref{eq:theta_bar_semi_bernoulli_alt}.
Note that $p = \Pr(\hat{E}^{(i)} > E_c) = \Pr(E^{(i)} > E_c)$.

By construction \eqref{eq:def_Ehat}, the first term is equivalent to
\begin{align*}
\hspace{20pt}
\lefteqn{\hspace{-20pt}\mathbb{E}\left[\left.C\left(\min\left\{\hat{E}^{(i)}-\tfrac{\bar{B}-\hat{E}^{(i)}}{T-1},\bar{B}\right\}\right)\right|\hat{E}^{(i)}>E_c\right]}\\*
&=\mathbb{E}\left[\left.C\left(\min\left\{E^{(i)}-\tfrac{\bar{B}-{E}^{(i)}}{T-1},\bar{B}\right\}\right)\right|{E}^{(i)}>E_c\right].
\end{align*}
For the second term, we use \eqref{eq:prob_Ehat_positive} followed by \eqref{eq:def_q} to obtain
\begin{align*}
C(E_c\cdot\Pr(\hat{E}^{(i)}\geq E_c))
&= C(E_c\cdot q)\\
&= C(\mathbb{E}[\min(E^{(i)},E_c)]).
\end{align*}
It follows that
\begin{align}
\mathscr{T}(\mathbf{g}) &\geq \tfrac{p(T-1)}{T}\mathbb{E}\Big[C\big(\min\big\{{E}^{(i)}-\tfrac{\bar{B}-{E}^{(i)}}{T-1},\bar{B}\big\}\big)\Big|{E}^{(i)}>E_c\Big]\nonumber\\*
 &\qquad + \tfrac{p+T(1-p)}{T}C(\mathbb{E}[\min(E^{(i)},E_c)]) - \tfrac{1}{2}\log e,
\end{align}
which concludes the proof of the lower bound in Theorem~\ref{thm:throughput_bounds}.

\section{Upper Bound}
\label{sec:upper_bound}

Fix an arbitrary policy $\mathbf{g}$, and consider the expected total throughput for $N$ blocks, or equivalently $NT$ time slots:
\begin{align}
\mathscr{T}_{NT}(\mathbf{g})
&= \sum_{t=1}^{NT}\mathbb{E}[C(g_t)]\nonumber\\
&= \sum_{i=1}^{N}\sum_{j=1}^{T}\mathbb{E}[C(g_j^{(i)})]\nonumber\\
&= \sum_{j=1}^{T-1}\mathbb{E}[C(g_j^{(1)})]
	+\mathbb{E}[C(g_T^{(N)})]\nonumber\\*
	&\qquad
	+\sum_{i=2}^{N}\bigg(\sum_{j=1}^{T-1}\mathbb{E}[C(g_j^{(i)})]
	+\mathbb{E}[C(g_T^{(i-1)})]\bigg)\nonumber\\
&\leq \sum_{i=2}^{N}\bigg(\sum_{j=1}^{T-1}\mathbb{E}[C(g_j^{(i)})]
	+\mathbb{E}[C(g_T^{(i-1)})]\bigg)
	+ TC(\bar{B}),
\end{align}
where the inequality is because $g_j^{(i)} \leq b_j^{(i)} \leq \bar{B}$ for any $i,j$.
According to Lemma~\ref{lemma:MDP_reduction} (specifically \eqref{eq:reward_upper_bound} and \eqref{eq:first_T-1_powers_upper_bound} in Appendix~\ref{sec:MDP_reduction}), we can upper bound the total rate of the first $T-1$ time slots in each block as follows:
\begin{align*}
\mathscr{T}_{NT}(\mathbf{g})
&\leq \sum_{i=2}^{N}\mathbb{E}\Big[(T-1)C\big(\min\big\{E^{(i)}-\tfrac{b_T^{(i)}-b_1^{(i)}}{T-1}, \bar{B}\big\}\big)\Big]\\*
&\qquad+\sum_{i=2}^{N}\mathbb{E}\big[C(g_T^{(i-1)})\big] + TC(\bar{B}).
\end{align*}
Since $b_1^{(i)} = \min\{b_T^{(i-1)} - g_T^{(i-1)} + E^{(i)}, \bar{B}\}$ and $C(\,\cdot\,)$ is nondecreasing:
\begin{align}
\mathscr{T}_{NT}(\mathbf{g})
&\leq \sum_{i=2}^{N}\theta_i + TC(\bar{B}),\label{eq:total_throughput_upper_bound}
\end{align}
where we have denoted for $i=2,\ldots,N$:
\begin{align*}
\theta_i &= \mathbb{E}\Big[(T-1)C\big(\min\big\{E^{(i)} - \tfrac{b_T^{(i)}-b_T^{(i-1)}+g_T^{(i-1)}-E^{(i)}}{T-1}, \bar{B}\big\}\big)\Big]\\
&\qquad+\mathbb{E}\big[C(g_T^{(i-1)})\big].
\end{align*}
We break the expectation over $E^{(i)}$ as shown in equations \eqref{eq:theta_i_ub_i}--\eqref{eq:theta_i_upper_bound} at the top of the next page,
where \eqref{eq:theta_i_ub_i} is because $\Pr(E^{(i)}\leq E_c) = 1-p$;
\eqref{eq:theta_i_ub_ii} is due to the monotonicity of $C(\,\cdot\,)$;
\eqref{eq:theta_i_ub_iii} is because $g_T^{(i-1)}$ is independent of $E^{(i)}$ (by definition of an online policy and because $E^{(i)}$ are i.i.d.), and by denoting 
\begin{equation*}
p' = \Pr(E^{(i)} > \bar{B}),
\end{equation*}
and \eqref{eq:theta_i_upper_bound} is by applying Jensen's inequality to the two terms in the first expectation.
\addtocounter{equation}{4}

Denote the following expected values, for $i=1,\ldots,N$:
\begin{align*}
\gamma^{(i)} &= \mathbb{E}[g_T^{(i)}],\\
\tilde{\mu} &= \mathbb{E}[E^{(i)} | E^{(i)}\leq E_c],\\
\beta^{(i)} &= \mathbb{E}[b_T^{(i)}],\\
\beta_0^{(i)} &= \mathbb{E}[b_T^{(i)} | E^{(i)} \leq E_c],\\
\beta_1^{(i)}(x) &= \mathbb{E}[b_T^{(i)} | E^{(i)} = x],\quad E_c < x \leq \bar{B}.
\end{align*}

\begin{figure*}[!t]
\normalsize
\setcounter{mytempeqncnt}{\value{equation}}
\setcounter{equation}{53}
\begin{align}
\theta_i &\overset{\text{(i)}}{=} \mathbb{E}[C(g_T^{(i-1)})]
	+ (1-p)(T-1)\mathbb{E}\Big[C\big(\min\big\{E^{(i)}-\tfrac{b_T^{(i)}-b_T^{(i-1)}+g_T^{(i-1)}-E^{(i)}}{T-1},\bar{B}\big\}\big)\Big|E^{(i)}\leq E_c\Big]\nonumber\\*
	&\qquad+\sum_{x>E_c}P_E(x)(T-1)\mathbb{E}\Big[C\big(\min\big\{x-\tfrac{b_T^{(i)}-b_T^{(i-1)}+g_T^{(i-1)}-x}{T-1},\bar{B}\big\}\big)\Big|E^{(i)}=x\Big]
	\label{eq:theta_i_ub_i}\\
&\overset{\text{(ii)}}{\leq} \mathbb{E}[C(g_T^{(i-1)})]
	+ (1-p)(T-1)\mathbb{E}\Big[C\big(E^{(i)}-\tfrac{b_T^{(i)}-b_T^{(i-1)}+g_T^{(i-1)}-E^{(i)}}{T-1}\big)\Big|E^{(i)}\leq E_c\Big]\nonumber\\*
	&\qquad +\sum_{E_c<x\leq\bar{B}}P_E(x)(T-1)\mathbb{E}\Big[C\big(x-\tfrac{b_T^{(i)}-b_T^{(i-1)}+g_T^{(i-1)}-x}{T-1}\big)\Big|E^{(i)}=x\Big]+\Pr(E^{(i)}>\bar{B})\cdot(T-1)C(\bar{B})
	\label{eq:theta_i_ub_ii}\\
&\overset{\text{(iii)}}{=} \mathbb{E}\Big[C(g_T^{(i-1)}) + (1-p)(T-1)
	C\big(E^{(i)}-\tfrac{b_T^{(i)}-b_T^{(i-1)}+g_T^{(i-1)}-E^{(i)}}{T-1}\big)
	\Big|E^{(i)}\leq E_c\Big]\nonumber\\*
	&\qquad +\sum_{E_c<x\leq\bar{B}}P_E(x)(T-1)\mathbb{E}\Big[C\big(x-\tfrac{b_T^{(i)}-b_T^{(i-1)}+g_T^{(i-1)}-x}{T-1}\big)\Big|E^{(i)}=x\Big]+p'(T-1)C(\bar{B})
	\label{eq:theta_i_ub_iii}\\
&\overset{\text{(iv)}}{\leq} (p+(1-p)T)\mathbb{E}\Big[C\big(\tfrac{pg_T^{(i-1)}+(1-p)(TE^{(i)}-b_T^{(i)}+b_T^{(i-1)}}{p+(1-p)T}\big)\Big|E^{(i)}\leq E_c\Big]\nonumber\\*
	&\qquad +\sum_{E_c<x\leq\bar{B}}P_E(x)(T-1)\mathbb{E}\Big[C\big(\tfrac{Tx-b_T^{(i)}+b_T^{(i-1)}-g_T^{(i-1)}}{T-1}\big)\Big|E^{(i)}=x\Big] + p'(T-1)C(\bar{B}),
	\label{eq:theta_i_upper_bound}
\end{align}
\setcounter{equation}{\value{mytempeqncnt}}
\hrulefill
\vspace*{4pt}
\end{figure*}

\begin{figure*}[!b]
\vspace*{4pt}
\hrulefill
\normalsize
\setcounter{mytempeqncnt}{\value{equation}}
\setcounter{equation}{59}
\begin{align}
\tfrac{1}{NT}\mathscr{T}_{NT}(\mathbf{g})
&\leq \tfrac{1}{N}C(\bar{B})
+ \tfrac{1}{N}\sum_{i=2}^{N}\tfrac{p+(1-p)T}{T}C\big(\tfrac{p\gamma^{(i-1)} + (1-p)(T\tilde{\mu}-\beta_0^{(i)}+\beta^{(i-1)})}{p+(1-p)T}\big)\nonumber\\*
&\qquad +\tfrac{1}{N}\sum_{i=2}^{N}\tfrac{T-1}{T}\sum_{E_c<x\leq\bar{B}}P_E(x)C\big(\tfrac{Tx-\beta_1^{(i)}(x)+\beta^{(i-1)}-\gamma^{(i-1)}}{T-1}\big)
+\tfrac{N-1}{N}p'\tfrac{T-1}{T}C(\bar{B}).
\label{eq:average_throughput_upper_bound}
\end{align}
\setcounter{equation}{\value{mytempeqncnt}}
\end{figure*}

Note that since $E^{(i)}>\bar{B}$ implies $b_T^{(i)}=\bar{B}$, the following relation holds:
\begin{equation}
\beta^{(i)} = (1-p) \beta_0^{(i)} + \sum_{E_c<x\leq\bar{B}}P_E(x)\beta_1^{(i)}(x) + p'\bar{B}.
\label{eq:battery_steady_state}
\end{equation}
Now, applying Jensen's inequality to \eqref{eq:theta_i_upper_bound}, and again observing that $g_T^{(i-1)}$, as well as $b_T^{(i-1)}$, are independent of $E^{(i)}$:
\begin{align}
\theta_i &\leq (p+(1-p)T) C\big(\tfrac{p\gamma^{(i-1)} + (1-p)(T\tilde{\mu}-\beta_0^{(i)}+\beta^{(i-1)})}{p+(1-p)T}\big)\nonumber\\*
&\qquad +(T-1)\sum_{E_c<x\leq\bar{B}}P_E(x)C\big(\tfrac{Tx-\beta_1^{(i)}(x)+\beta^{(i-1)}-\gamma^{(i-1)}}{T-1}\big)\nonumber\\*
&\qquad+p'(T-1)C(\bar{B}).
\end{align}
Substituting into \eqref{eq:total_throughput_upper_bound} and dividing by $NT$ yields \eqref{eq:average_throughput_upper_bound} at the bottom of the page.
\addtocounter{equation}{1}

Denote the following time-averages of the expected values above:
\begin{align*}
\bar{\gamma} &= \tfrac{1}{N-1}\sum_{i=1}^{N-1}\gamma^{(i)},\\
\bar{\beta} &= \tfrac{1}{N-1}\sum_{i=1}^{N-1}\beta^{(i)},\\
\bar{\beta}_0 &= \tfrac{1}{N-1}\sum_{i=1}^{N-1}\beta_0^{(i)},\\
\bar{\beta}_1(x) &= \tfrac{1}{N-1}\sum_{i=1}^{N-1}\beta_1^{(i)}(x),\quad E_c<x\leq\bar{B}.
\end{align*}
We again apply Jensen's inequality to \eqref{eq:average_throughput_upper_bound}, giving \eqref{eq:average_throughput_edge_effects1}--\eqref{eq:average_throughput_upper_bound2} at the top of the next page.
Eq. \eqref{eq:average_throughput_edge_effects2} follows because $0\leq b_T^{(i)}\leq \bar{B}$ for any $i$;
and \eqref{eq:average_throughput_upper_bound2}
is due to the fact that $\ln(1+x+y) \leq \ln(1+x) + y$ for $x,y\geq0$,
or equivalently $C(x+y) \leq C(x) + \tfrac{1}{2\ln 2}y$.
\addtocounter{equation}{3}

Since $g_T^{(i)}\leq b_T^{(i)}$, we can take expectation to obtain $\gamma^{(i)}\leq\beta^{(i)}$ for every $i$, which implies $\bar{\gamma}\leq\bar{\beta}$.
Similarly, $b_T^{(i)}\leq\bar{B}$, which implies $\beta_1^{(i)}(x)\leq\bar{B}$, and consequently $\bar{\beta}_1(x) \leq \bar{B}$ for all $E_c<x\leq\bar{B}$.
Finally, the relation \eqref{eq:battery_steady_state} implies 
\[
\bar{\beta} = (1-p)\bar{\beta}_0 + \sum_{E_c<x\leq\bar{B}}P_E(x)\bar{\beta}_1(x) + p'\bar{B}.
\]
With these constraints, we can further upper bound \eqref{eq:average_throughput_upper_bound2} as follows:
\begin{equation}
\tfrac{1}{NT}\mathscr{T}_{NT}(\mathbf{g})
\leq \tfrac{N-1}{N}\Theta^\star + \tfrac{1}{N}C(\bar{B}) + \tfrac{1}{2\ln2}\tfrac{1}{NT}\bar{B},
\label{eq:average_throughput_upper_bound3}
\end{equation}
where $\Theta^\star$ is the optimal solution to the following convex optimization problem:\footnote{These constraints are only necessary but are not sufficient to describe the optimization variables; nevertheless, this is still a valid upper bound on \eqref{eq:average_throughput_upper_bound2}.}
\begin{equation}
\begin{aligned}
&\hspace{-6pt}\underset{\substack{\gamma,\,\beta,\,\beta_0,\\ \{\beta_1(x)\}_{E_c<x\leq\bar{B}}}}{\text{maximize}}
&&\tfrac{p+(1-p)T}{T}C\big(\tfrac{p\gamma+(1-p)(T\tilde{\mu}-\beta_0+\beta)}{p+(1-p)T}\big)+p'\tfrac{T-1}{T}C(\bar{B})\\[-10pt]
&&&\quad
+\tfrac{T-1}{T}\sum_{E_c<x\leq\bar{B}}P_E(x)C\big(\tfrac{Tx-\beta_1(x)+\beta-\gamma}{T-1}\big)\\
&\text{subject to}
&&\gamma\leq\beta,\\
&&&\beta_1(x)\leq\bar{B},\qquad E_c<x\leq\bar{B},\\
&&&\beta=(1-p)\beta_0+\sum_{E_c<x\leq\bar{B}}P_E(x)\beta_1(x)+p'\bar{B}.
\end{aligned}
\label{eq:opt_problem}
\end{equation}
\begin{figure*}[!t]
\normalsize
\setcounter{mytempeqncnt}{\value{equation}}
\setcounter{equation}{60}
\begin{align}
\tfrac{1}{NT}\mathscr{T}_{NT}(\mathbf{g})
&\leq \tfrac{1}{N}C(\bar{B}) 
+\tfrac{N-1}{N}\tfrac{p+(1-p)T}{T}C\bigg(\frac{p\bar{\gamma}+(1-p)(T\tilde{\mu}-\bar{\beta}_0-\tfrac{1}{N-1}(\beta_0^{(N)}-\beta_0^{(1)})+\bar{\beta})}{p+(1-p)T}\bigg)\nonumber\\*
&\quad + \tfrac{N-1}{N}\tfrac{T-1}{T}\sum_{E_c<x\leq\bar{B}}P_E(x)C\bigg(\frac{Tx-\bar{\beta}_1(x)-\tfrac{1}{N-1}(\beta_1^{(N)}(x)-\beta_1^{(1)}(x))+\bar{\beta}-\bar{\gamma}}{T-1}\bigg)
+\tfrac{N-1}{N}p'\tfrac{T-1}{T}C(\bar{B})
\label{eq:average_throughput_edge_effects1}\\
&\leq \tfrac{1}{N}C(\bar{B}) 
+\tfrac{N-1}{N}\tfrac{p+(1-p)T}{T}C\bigg(\frac{p\bar{\gamma}+(1-p)(T\tilde{\mu}-\bar{\beta}_0+\bar{\beta}+\tfrac{1}{N-1}\bar{B})}{p+(1-p)T}\bigg)\nonumber\\*
&\quad + \tfrac{N-1}{N}\tfrac{T-1}{T}\sum_{E_c<x\leq\bar{B}}P_E(x)C\bigg(\frac{Tx-\bar{\beta}_1(x)+\bar{\beta}-\bar{\gamma}+\tfrac{1}{N-1}\bar{B}}{T-1}\bigg)
+\tfrac{N-1}{N}p'\tfrac{T-1}{T}C(\bar{B})
\label{eq:average_throughput_edge_effects2}\\
&\leq \tfrac{1}{N}C(\bar{B})
+\tfrac{N-1}{N}\tfrac{p+(1-p)T}{T}C\big(\tfrac{p\bar{\gamma}+(1-p)(T\tilde{\mu}-\bar{\beta}_0+\bar{\beta})}{p+(1-p)T}\big)
+\tfrac{1}{2\ln 2}\tfrac{1-p}{NT}\bar{B}\nonumber\\*
&\quad +\tfrac{N-1}{N}\tfrac{T-1}{T}\sum_{E_c<x\leq\bar{B}}P_E(x)C\big(\tfrac{Tx-\bar{\beta}_1(x)+\bar{\beta}-\bar{\gamma}}{T-1}\big)
+\tfrac{1}{2\ln 2}\tfrac{p}{NT}\bar{B}
+\tfrac{N-1}{N}p'\tfrac{T-1}{T}C(\bar{B}).
\label{eq:average_throughput_upper_bound2}
\end{align}
\setcounter{equation}{\value{mytempeqncnt}}
\hrulefill
\vspace*{4pt}
\end{figure*}
It is shown in Appendix~\ref{sec:kkt_verification}, by verifying that KKT conditions hold, that the following solution is optimal:
\begin{equation}
\gamma^\star=\beta^\star=p\bar{B},\qquad
\beta_0^\star=0,\qquad
\beta_1^\star(x)=\bar{B},\ E_c<x\leq\bar{B}.
\label{eq:opt_solution}
\end{equation}
The optimal objective is given by
\begin{align}
\Theta^\star
&=\tfrac{p+(1-p)T}{T}C\big(\tfrac{p\bar{B}+(1-p)T\tilde{\mu}}{p+(1-p)T}\big)\nonumber\\*
&\quad\ 
+\tfrac{T-1}{T}\sum_{E_c<x\leq\bar{B}}P_E(x)C\big(\tfrac{Tx-\bar{B}}{T-1}\big)
+p'\tfrac{T-1}{T}C(\bar{B})\nonumber\\
&\overset{\text{(i)}}{=}\tfrac{p+(1-p)T}{T}C(\mathbb{E}[\min(E^{(i)},\,E_c)])\nonumber\\*
&\quad\
+\tfrac{T-1}{T}p\mathbb{E}\big[C\big(\min\big\{E^{(i)}-\tfrac{\bar{B}-E^{(i)}}{T-1},\bar{B}\big\}\big)\big|E^{(i)}>E_c\big]\nonumber\\
&\overset{\text{(ii)}}{=}\bar{\Theta},
\end{align}
where (i) is by \eqref{eq:def_q_alt} and \eqref{eq:def_q},
and (ii) is by definition of the approximate throughput \eqref{eq:theta_bar_general}.
Substituting in \eqref{eq:average_throughput_upper_bound3} and taking the limit as $N\to\infty$:
\begin{align*}
\liminf_{N\to\infty}\tfrac{1}{NT}\mathscr{T}_{NT}(\mathbf{g})
&\leq \bar{\Theta}.
\end{align*}
Since the policy $\mathbf{g}$ was arbitrary, this implies $\Theta \leq \bar{\Theta}$.
This concludes the proof of the upper bound in Theorem~\ref{thm:throughput_bounds}.

\section{Conclusion}
\label{sec:conclusion}

We proposed a simple online power control policy for the energy harvesting channel with block i.i.d.\ energy arrivals, and showed that it is within a constant gap from the optimum.
This resulted in a simple and insightful formula that approximates the optimal throughput.
Previously, optimal power control has been characterized for the offline case and for the online case with i.i.d.\ energy arrivals. Our results reveal how correlation in the energy harvesting process impacts online power control and the corresponding optimal throughput. While in this paper we consider block i.i.d.\ energy arrivals with arbitrary distribution, the development of these results for an important special case, namely block i.i.d. Bernoulli arrivals, can be found in our preliminary work \cite{BlockiidBernoulli}, which provided the insights for the current paper. To the best of our knowledge, this is the first paper to develop online power control policies with explicit guarantees on optimality for an energy arrival process with memory and a finite battery.

\appendices

\section{Proof of Lemma~\ref{lemma:MDP_reduction}}
\label{sec:MDP_reduction}

Consider an arbitrary set of power allocations $g_1^{(i)},\ldots,g_{T}^{(i)}$ that satisfy the energy constraints \eqref{eq:power} and \eqref{eq:battery}. 
We will show that we can replace the first $T-1$ elements with an appropriate $\tilde{g}_1^{(i)},\ldots,\tilde{g}_{T-1}^{(i)}$ such that $\tilde{g}_1^{(i)}=\ldots=\tilde{g}_{T-1}^{(i)}$, while still preserving the energy constraints and without decreasing the reward.

Specifically, let $b_T^{(i)}$ be the battery state at the last time slot of block $i$ given $g_1^{(i)},\ldots,g_{T-1}^{(i)}$. Set
\[
\tilde{g}_1^{(i)}=\ldots=\tilde{g}_{T-1}^{(i)}
=\min\big(E^{(i)}-\tfrac{b_T^{(i)}-b_1^{(i)}}{T-1},\bar{B}\big).
\]
We will show that this is an admissible policy, which produces a higher reward than the original policy, and results in the same battery state $b_T^{(i)}$ at the end of the block.

First, the reward obtained by the original policy can be upper bounded using concavity of the function $C(\,\cdot\,)$:
\begin{align}
r_i&=\frac{1}{T}\sum_{j=1}^{T}C(g_j^{(i)})\nonumber\\
&\leq \frac{T-1}{T}C\Big(\frac{1}{T-1}\sum_{j=1}^{T-1}g_j^{(i)}\Big)+\frac{1}{T}C(g_T^{(i)}).\label{eq:reward_upper_bound}
\end{align}
It follows from \eqref{eq:battery} that $b_j^{(i)}\leq b_{j-1}^{(i)}-g_{j-1}^{(i)}+E^{(i)}$ for $j=2,\ldots,T-1$, which implies 
\begin{equation*}
b_T^{(i)}\leq b_1^{(i)}+(T-1)E^{(i)}-\sum_{j=1}^{T-1}g_j^{(i)}.
\end{equation*}
Additionally, we clearly have $g_j^{(i)}\leq b_j^{(i)}\leq\bar{B}$ for every $j=1,\ldots,T-1$, thus
\begin{equation}
\frac{1}{T-1}\sum_{j=1}^{T-1}g_j^{(i)}\leq \min\big(E^{(i)}-\tfrac{b_T^{(i)}-b_1^{(i)}}{T-1}, \bar{B}\big)=\tilde{g}_1^{(i)}.
\label{eq:first_T-1_powers_upper_bound}
\end{equation}
Substituting in \eqref{eq:reward_upper_bound} and using the fact that $C(\,\cdot\,)$ is non-decreasing yields
\begin{align*}
r_i&\leq\frac{T-1}{T}C(\tilde{g}_1^{(i)})+\frac{1}{T}C(g_T^{(i)}).
\end{align*}

It remains to show that the policy $\tilde{g}_j^{(i)}$ is admissible, and that it results in the same final battery state $b_T^{(i)}$.
To this end, denote the battery state resulting from the new policy by $\tilde{b}_j^{(i)}$, $j=1,\ldots,T-1$. That is:
\begin{align*}
\tilde{b}_1^{(i)}&=b_1^{(i)},\\
\tilde{b}_j^{(i)}&=\min\{b_{j-1}^{(i)}-\tilde{g}_1^{(i)}+E^{(i)},\bar{B}\},
\qquad j=2,\ldots,T-1.
\end{align*}

If $E^{(i)}\geq\bar{B}$, then $b_1^{(i)}=b_T^{(i)}=\bar{B}$ regardless of the policy $g_1^{(i)},\ldots,g_{T-1}^{(i)}$, which implies $\tilde{g}_1^{(i)}=\bar{B}$ and $\tilde{b}_j^{(i)}=\bar{B}$ for all $j=1,\ldots,T-1$.
It trivially follows that the policy is admissible and $\tilde{b}_T^{(i)}=b_T^{(i)}$.

Assume therefore $E^{(i)}<\bar{B}$.
Under this assumption, it is easy to see that $E^{(i)}\leq b_j^{(i)}\leq\bar{B}$ for all $j$, hence
\[
E^{(i)}-\tfrac{b_T^{(i)}-b_1^{(i)}}{T-1}
\leq \tfrac{(T-2)E^{(i)}+b_1^{(i)}}{T-1}
< \bar{B},
\]
which implies $\tilde{g}_1^{(i)}=E^{(i)}-\tfrac{b_T^{(i)}-b_1^{(i)}}{T-1}$.
We will show by induction that
\begin{equation}
\tilde{b}_j^{(i)}=\frac{(T-j)b_1^{(i)}+(j-1)b_T^{(i)}}{T-1},
\quad j=1,\ldots,T-1.
\label{eq:new_battery_state}
\end{equation}
This is clearly true for $j=1$. Assuming it holds for $j$, we have for $j+1$:
\begin{align*}
\tilde{b}_{j+1}^{(i)}
&=\min\big\{\tilde{b}_j^{(i)}-\tilde{g}_1^{(i)}+E^{(i)}, \bar{B}\big\}\\
&=\min\Big\{\tfrac{(T-j)b_1^{(i)}+(j-1)b_T^{(i)}}{T-1}
	-E^{(i)}+\tfrac{b_T^{(i)}-b_1^{(i)}}{T-1}+E^{(i)}, \bar{B}\Big\}\\
&=\tfrac{(T-j-1)b_1^{(i)}+j b_T^{(i)}}{T-1},
\end{align*}
where in the last step we used the fact that $b_1^{(i)},b_T^{(i)}\leq\bar{B}$.

It is now clear that $\tilde{b}_T^{(i)}=b_T^{(i)}$.
To see that the policy $\tilde{g}_j^{(i)}$ is admissible, recall that
$\tilde{g}_1^{(i)}\leq \tfrac{(T-2)E^{(i)}+b_1^{(i)}}{T-1}$,
which implies $\tilde{g}_1^{(i)}\leq b_1^{(i)}=\tilde{b}_1^{(i)}$ and $\tilde{g}_1^{(i)}\leq\tfrac{(T-2)b_T^{(i)}+b_1^{(i)}}{T-1}=\tilde{b}_{T-1}^{(i)}$, which in turn implies $\tilde{g}_1^{(i)}\leq\tilde{b}_j^{(i)}$ for all $j=1,\ldots,T-1$.

Note that since $b_T^{(i)}\geq E^{(i)}$, we have $\tilde{g}_1^{(i)}\leq E^{(i)}-\tfrac{E^{(i)}-b_1^{(i)}}{T-1}$.
We conclude that we can reduce the action space to the set of all policies for which $g_1^{(i)}=\ldots=g_{T-1}^{(i)}$ and
\[
0\leq g_1^{(i)}\leq \min\big(E^{(i)}-\tfrac{E^{(i)}-b_1^{(i)}}{T-1},\bar{B}\big).
\]
This concludes the proof of the lemma.\qed

\section{Equation \eqref{eq:def_Ec} has a Unique Solution}
\label{sec:unique_Ec}

Consider the function
\[
f(x) = \bar{B} - Tx + (T-1)\mathbb{E}[\min(E^{(i)}, x)].
\]
Observe that $f(0)=\bar{B}>0$ and
$f(\bar{B}) = (T-1)(\mathbb{E}[\min(E^{(i)},\bar{B})-\bar{B}) \leq 0$.
Moreover, it follows from the dominated convergence theorem \cite[Thm. 1.5.6]{durrett2010probability} that $\mathbb{E}[\min(E^{(i)},x)]$ is a continuous function of $x$, and therefore $f(x)$ is continuous as well.
It then follows that $f(x)$ must have a zero in the interval $[0,\bar{B}]$.

To show this zero is unique, we will show that $f(x)$ is monotonically decreasing. Since
$f(x)=\bar{B} - x - (T-1)(x - \mathbb{E}[\min(E^{(i)},x)])$,
it is enough to show that the function $h(x)=x - \mathbb{E}[\min(E^{(i)},x)]$ is monotonically increasing. Indeed, for any $0\leq x< y$:
\begin{align*}
h(x)&=x-\mathbb{E}[E^{(i)}\cdot 1\{E^{(i)}\leq x\}]-x \Pr(E^{(i)}>x)\\
&=x\Pr(E^{(i)}\leq x)-\mathbb{E}[E^{(i)}\cdot 1\{E^{(i)}\leq x\}]\\
&=\mathbb{E}[(x-E^{(i)})\cdot 1\{E^{(i)}\leq x\}]\\
&<\mathbb{E}[(y-E^{(i)})\cdot 1\{E^{(i)}\leq y\}]\\
&=h(y).
\end{align*}
Hence $f(x)$ is strictly decreasing, and it must have a unique zero in the interval $[0,\bar{B}]$, which is the unique solution to~\eqref{eq:def_Ec}.\qed

\section{Proof of Proposition~\ref{prop:capacity_throughput_connection}}
\label{sec:inf_theory_gap}

The proof follows closely the proof of Theorem 4 in~\cite{ShavivNguyenOzgur2016}. We provide a sketch of the proof, elaborating only the necessary modifications for block i.i.d.\ energy arrivals.

In addition to the capacity $C$, we will also be interested in the capacity when energy arrival observations are available at the receiver, i.e. the decoding function is of the form $f^{\mathrm{dec}}:\mathcal{Y}^N\times\mathcal{E}^N\to\mathcal{M}$. Denote this capacity by $C_{\mathrm{Rx}}$. 

Clearly $C \leq C_{\mathrm{Rx}}$.
On the other hand, by \cite[Thm.~1]{Jafar2006}, the maximum capacity improvement due to the availability of side information is limited by the amount of side information itself, that is $C_{\mathrm{Rx}} - C \leq \lim_{N\to\infty}\tfrac{1}{N}H(E^N)$ (see also \cite[Sec.~VI-B2]{ShavivNguyenOzgur2016}).
Due the block i.i.d.\ structure of the energy arrivals process, we get
\begin{equation}
C_{\mathrm{Rx}} - \tfrac{1}{T}H(E^{(i)}) \leq C \leq C_{\mathrm{Rx}}.
\label{eq:capacity_without_rx_side_information}
\end{equation}

Following the proof of \cite[Thm.~2]{ShavivNguyenOzgur2016}, it can be seen that $C_{\mathrm{Rx}}$ is given by the following $N$-letter expression:
\begin{align}
C_{\mathrm{Rx}} &= \lim_{N\to\infty}\frac{1}{NT}\sup I(X^{NT}; Y^{NT}| E^{(1)},\ldots,E^{(N)}).
\label{eq:capacity_n_letter}
\end{align}
where the supremum is over all input probability distributions of the form
\begin{align*}
\hspace{30pt}\lefteqn{\hspace{-30pt}p(x^{NT}|e^{(1)},\ldots,e^{(N)})}\nonumber\\*
&=\prod_{i=1}^{N}p(x_{(i-1)T+1}^{iT}|x^{(i-1)T},e^{(1)},\ldots,e^{(i)}),
\end{align*}
which satisfy the energy constraints \eqref{eq:inf_theory_amplitude} and \eqref{eq:inf_theory_battery} with probability 1 for any sequence of energy arrivals blocks $e^{(1)},\ldots,e^{(N)}$.
More precisely, for any $(e^{(1)},\ldots,e^{(N)})\in\mathcal{E}^N$, the input probability distribution defines an RV $X^{NT}$. This RV must satisfy 
\begin{align*}
X_t^2 &\leq B_t,\\
B_t &= \min\{B_{t-1} - X_{t-1}^2 + e^{(i)}, \bar{B}\}
\end{align*}
almost surely for all $t=(i-1)T+1,\ldots,iT$ and $i=1,\ldots,N$.

Next, we can derive an upper and lower bound on $C_{\mathrm{Rx}}$ following the same lines as the proof of \cite[Thm.~4]{ShavivNguyenOzgur2016}:
\begin{equation}
\Theta - \tfrac{1}{2}\log\big(\tfrac{\pi e}{2}\big) \leq C_{\mathrm{Rx}} \leq \Theta.
\label{eq:capacity_rx_throughput}
\end{equation}
It should be noted that this proof does not depend on the statistics of the energy arrivals process, and could be applied to any energy arrivals process as long as the underlying physical channel is memoryless and capacity is given by an $N$-letter mutual information expression such as \eqref{eq:capacity_n_letter}.
Combining \eqref{eq:capacity_without_rx_side_information} with \eqref{eq:capacity_rx_throughput} yields the desired result.\qed

\section{Proof of Proposition~\ref{prop:semi_bernoulli_is_worst}}
\label{sec:semi_bernoulli_is_worst}

We prove Proposition~\ref{prop:semi_bernoulli_is_worst} by induction.
Clearly for $N=1$ we have
\[
J_1(x,s)=\hat{J}_1(x,s)=r(x,s).
\]
Additionally, it follows from~\eqref{eq:reward} that $\hat{J}_1(x,s)$ is a concave and non-decreasing function of $x$ for $s=0,1$. This fact will be used in the induction proof.

Assume $J_{N-1}(x,s)\geq\hat{J}_{N-1}(x,s)$ for all $x\in[0,\bar{B}]$, $s\in\{0,1\}$, and also that $\hat{J}_{N-1}(x,s)$ is monotonic non-decreasing and concave in $x$ for $s=0,1$.

For the induction step, we write $J_N(x,s)$ as follows:
\begin{align}
J_N(x,s)&=r(x,s)+\mathbb{E}\big[J_{N-1}(x_2,s_2)\big]\nonumber\\
&\geq r(x,s)+\mathbb{E}\big[\hat{J}_{N-1}(x_2,s_2)\big],
\label{eq:theta_initial_expansion}
\end{align}
where $x_2$ and $s_2$ are given by \eqref{eq:x_state_dynamics} and \eqref{eq:s_state_dynamics}, with $x_1=x$ and $s_1=s$,
and the inequality is due to the induction assumption.
We further expand the second term:
\begin{align}
\mathbb{E}\big[\hat{J}_{N-1}(x_2,s_2)\big]
&=(1-p)\mathbb{E}\big[\hat{J}_{N-1}(x_2,0)\big|E^{(2)}\leq E_c\big]\nonumber\\*
&\quad\ 
+p\, \mathbb{E}\big[\hat{J}_{N-1}(x_2,1)\big|E^{(2)}>E_c\big].
\label{eq:theta_prev_expansion}
\end{align}
From \eqref{eq:x_state_dynamics}, we can succinctly write $x_2$ as follows:
\begin{equation}
x_2=\min\left\{(1-q)\tilde{x}+E^{(2)},\,
\tfrac{\bar{B}+(T-1)E^{(2)}}{T}\right\},
\end{equation}
where
\[
\tilde{x}=\begin{cases}
	\min(x,E_c), &\text{ if }s=0,\\
	E_c, &\text{ if }s=1.
\end{cases}
\]
We start with the first term of~\eqref{eq:theta_prev_expansion}.
Since $\hat{J}_{N-1}(z,0)$ is a non-decreasing function of $z$ by the induction assumption, it follows that
\begin{align*}
\hat{J}_{N-1}(x_2,0)
=\min\Big\{&
\hat{J}_{N-1}\big((1-q)\tilde{x}+E^{(2)},\,0\big),\nonumber\\*
&\qquad
\hat{J}_{N-1}\left(\tfrac{\bar{B}+(T-1)E^{(2)}}{T},\,0\right)\Big\}.
\end{align*}
Additionally, the functions 
\begin{align*}
f_1(z)&\triangleq\hat{J}_{N-1}\big((1-q)\tilde{x}+z,\,0\big),\\
f_2(z)&\triangleq\hat{J}_{N-1}\left(\tfrac{\bar{B}+(T-1)z}{T},\,0\right),
\end{align*}
are concave by the induction assumption that $\hat{J}_{N-1}(z,0)$ is concave.
It follows that the function $f_3(z)\triangleq\min\{f_1(z),f_2(z)\}$ is concave, as a minimum of two concave functions.
We conclude that $\hat{J}_{N-1}(x_2,0)=f_3(E^{(2)})$ is a concave function of $E^{(2)}$.

We make use of the following lemma, the proof of which can be found in~\cite{NearOptimal1}:
\begin{lemma}
\label{lemma:bernoulli_is_worst}
Let $f(z)$ be concave on the interval $[0,\bar{z}]$, and let $Z$ be a RV confined to the same interval, i.e. $0\leq Z\leq \bar{z}$. 
Let $\hat{Z}\in\{0,\bar{z}\}$ be a Bernoulli RV with $\Pr(\hat{Z}=\bar{z})=\mathbb{E}Z/\bar{z}$.
Then
\[
\mathbb{E}[f(Z)]\geq \mathbb{E}[f(\hat{Z})].
\]
\end{lemma}

Applying Lemma~\ref{lemma:bernoulli_is_worst} to the first term of \eqref{eq:theta_prev_expansion}, or equivalently $f_3(E^{(2)})$:
\begin{align}
\mathbb{E}\big[\hat{J}_{N-1}(x_2,0)\big|E^{(2)}\leq E_c\big]
&=\mathbb{E}\big[f_3(E^{(2)})\big|E^{(2)}\leq E_c\big]\nonumber\\
&\overset{\text{(i)}}{\geq}\mathbb{E}\big[f_3(W)\big]\nonumber\\
&\overset{\text{(ii)}}{=}\mathbb{E}\big[f_3(\hat{E}^{(2)})\big|\hat{E}^{(2)}\leq E_c\big]\nonumber\\
&=\mathbb{E}\big[\hat{J}_{N-1}(\hat{x}_2,0)\big|\hat{E}^{(2)}\leq E_c\big],
\label{eq:theta_expansion_first_term}
\end{align}
where (i) is by Lemma~\ref{lemma:bernoulli_is_worst} and \eqref{eq:def_W}, and (ii) is due to the construction~\eqref{eq:def_Ehat}.

Next, observe that by construction of $\hat{E}^{(i)}$ \eqref{eq:def_Ehat}, we have for any function $f(z)$:
\[
\mathbb{E}\big[f(E^{(i)})\big|E^{(i)}>E_c\big]
=\mathbb{E}\big[f(\hat{E}^{(i)})\big|\hat{E}^{(i)}>E_c\big].
\]
This implies that the second term of~\eqref{eq:theta_prev_expansion} is equal to
\begin{equation}
\mathbb{E}\big[\hat{J}_{N-1}(x_2,1)\big|E^{(2)}>E_c\big]
=\mathbb{E}\big[\hat{J}_{N-1}(\hat{x}_2,1)\big|\hat{E}^{(2)}>E_c\big].
\label{eq:theta_expansion_second_term}
\end{equation}
Substituting~\eqref{eq:theta_expansion_first_term} and~\eqref{eq:theta_expansion_second_term} in~\eqref{eq:theta_prev_expansion} yields
\begin{align}
\mathbb{E}\big[\hat{J}_{N-1}(x_2,s_2)\big]
&\geq (1-p)\mathbb{E}\big[\hat{J}_{N-1}(\hat{x}_2,0)\big|\hat{E}^{(2)}\leq E_c\big]\nonumber\\*
&\quad\ +p\mathbb{E}\big[\hat{J}_{N-1}(\hat{x}_2,1)\big|\hat{E}^{(2)}>E_c\big]\nonumber\\
&=\mathbb{E}\big[\hat{J}_{N-1}(\hat{x}_2,\hat{s}_2)\big].
\end{align}
Substituting this in~\eqref{eq:theta_initial_expansion}:
\begin{align}
J_N(x,s)&\geq r(x,s)+\mathbb{E}\big[\hat{J}_{N-1}(\hat{x}_2,\hat{s}_2)\big]\nonumber\\
&=\hat{J}_N(x,s).
\end{align}

It remains to show $\hat{J}_N(x,s)$ is concave and non-decreasing in $x$ for $s=0,1$.
Starting with $s=1$, we have:
\begin{align}
\hat{J}_N(x,1)
&=r(x,1)+\mathbb{E}\big[\hat{J}_{N-1}(\hat{x}_2,\hat{s}_2)\big],
\label{eq:Jhat_large_energy}
\end{align}
where $\hat{x}_2=\min\left\{(1-q)E_c+\hat{E}^{(2)}, \tfrac{\bar{B}+(T-1)\hat{E}^{(2)}}{T}\right\}$ by \eqref{eq:x_state_dynamics}.
Note that the second term of \eqref{eq:Jhat_large_energy} does not depend on $x$. From~\eqref{eq:reward}, the function $r(x,1)$ is concave and non-decreasing in $x$. Therefore, $\hat{J}_N(x,1)$ is concave and non-decreasing in $x$.

Next, for $s=0$, writing $\hat{J}_N(x,0)$ explicitly and expanding the expectation as before:
\begin{align}
\hat{J}_N(x,0)
&=r(x,0)+(1-p)\mathbb{E}\big[\hat{J}_{N-1}(\hat{x}_2,0)\big|\hat{E}^{(2)}\leq E_c\big]\nonumber\\*
&\quad\ +p\,\mathbb{E}\big[\hat{J}_{N-1}(\hat{x}_2,1)\big|\hat{E}^{(2)}>E_c\big],
\end{align}
where now 
\[\hat{x}_2=\min\Big\{(1-q)\min(x,E_c)+\hat{E}^{(2)},\,\tfrac{\bar{B}+(T-1)\hat{E}^{(2)}}{T}\Big\}.\]
For any $s\in\{0,1\}$, by the induction assumption of monotonicity of $\hat{J}_{N-1}(x,s)$:
\begin{align*}
\hat{J}_{N-1}(\hat{x}_2,s)
=\min\Big\{&\hat{J}_{N-1}\big((1-q)\min(x,E_c)+\hat{E}^{(2)},s\big),\nonumber\\*
&\quad\hat{J}_{N-1}\left(\tfrac{\bar{B}+(T-1)\hat{E}^{(2)}}{T},s\right)\Big\}.
\end{align*}
For fixed $\hat{E}^{(2)}$, the function
\[h_1(x)\triangleq\hat{J}_{N-1}\big((1-q)\min(x,E_c)+\hat{E}^{(2)},s)\]
 is concave and non-decreasing in $x$,
whereas $h_2\triangleq\hat{J}_{N-1}\left(\frac{\bar{B}+(T-1)\hat{E}^{(2)}}{T},s\right)$ is simply a constant.
Hence, $h_3(x)\triangleq\min\{h_1(x),h_2\}$ is concave and non-decreasing in $x$, that is $\hat{J}_{N-1}(\hat{x}_2,s)$ is a concave non-decreasing function of $x$ (where the dependence on $x$ is implicit in $\hat{x}_2$).
Taking expectation over $\hat{E}^{(2)}$ preserves concavity and monotonicity, hence the functions
\[
\mathbb{E}\big[\hat{J}_{N-1}(\hat{x}_2,0)\big|\hat{E}^{(2)}\leq E_c\big],\ 
\mathbb{E}\big[\hat{J}_{N-1}(\hat{x}_2,1)\big|\hat{E}^{(2)}> E_c\big],
\]
are both concave and non-decreasing in $x$.
Finally, since $r(x,0)$ is concave and non-decreasing by~\eqref{eq:reward}, we conclude that $\hat{J}_{N}(x,0)$ is concave and non-decreasing as a non-negative weighted sum of concave non-decreasing functions.
This concludes the proof by induction.\qed

\section{Solution to Optimization Problem~\eqref{eq:opt_problem}}
\label{sec:kkt_verification}

We will show that~\eqref{eq:opt_solution} is the optimal solution to optimization problem \eqref{eq:opt_problem}
by writing the Lagrangian:
\begin{align}
\mathscr{L}&=\tfrac{p+(1-p)T}{T}C\big(\tfrac{p\gamma+(1-p)[T\tilde{\mu}-\beta_0+\beta]}{p+(1-p)T}\big)\nonumber\\*
&\quad+\tfrac{T-1}{T}\sum_{E_c<x\leq\bar{B}}P_E(x)C\big(\tfrac{Tx-\beta_1(x)+\beta-\gamma}{T-1}\big)+\tfrac{T-1}{T}p'C(\bar{B})\nonumber\\*
&\quad +\lambda_0(\beta-\gamma)+\sum_{E_c<x\leq\bar{B}}\lambda_1(x)\big(\bar{B}-\beta_1(x)\big)\nonumber\\*
&\quad+\nu\Big(\beta-(1-p)\beta_0-\sum_{E_c<x\leq\bar{B}}P_E(x)\beta_1(x)-p'\bar{B}\Big),
\end{align}
and verifying KKT conditions hold, namely the dual variables are non-negative: $\lambda_0\geq0$ and $\lambda_1(x)\geq0$ for every $E_c<x\leq\bar{B}$, and the gradient vanishes at the point $\gamma^\star$, $\beta^\star$, $\beta_0^\star$, $\{\beta_1^\star(x)\}_{E_c<x\leq\bar{B}}$.

To simplify calculations, we normalize by $\frac{1}{2}\log e$ without loss of generality, so that $C(x)=\ln(1+x)$.
We start by taking the derivative with respect to $\gamma$, and substituting the values given in~\eqref{eq:opt_solution}:
\begin{align}
\left.\frac{\partial\mathscr{L}}{\partial\gamma}\right|_{\gamma^\star,\beta^\star,\beta_0^\star,\beta_1^\star}
&=\tfrac{p+(1-p)T}{T}\tfrac{p}{p+(1-p)T+p\bar{B}+(1-p)T\tilde{\mu}}\nonumber\\*
&\quad-\tfrac{T-1}{T}\sum_{E_c<x\leq\bar{B}}\tfrac{P_E(x)}{T-1+Tx-\bar{B}}-\lambda_0.
\end{align}
Since the gradient must vanish, we get:
\begin{align}
\lambda_0&=\tfrac{p+(1-p)T}{T}\tfrac{p}{p+(1-p)T+p\bar{B}+(1-p)T\tilde{\mu}}\nonumber\\*
&\quad\ -\tfrac{T-1}{T}\sum_{E_c<x\leq\bar{B}}\tfrac{P_E(x)}{T-1+Tx-\bar{B}}\label{eq:lambda0}\\
&>\tfrac{p+(1-p)T}{T}\tfrac{p}{p+(1-p)T+p\bar{B}+(1-p)T\tilde{\mu}}
-\tfrac{T-1}{T}\tfrac{p-p'}{T-1+TE_c-\bar{B}},
\label{eq:lambda0_lower_bound}
\end{align}
where $p-p'=\Pr(E_c<E^{(i)}\leq\bar{B})$.
Observe that from~\eqref{eq:def_q_alt} and \eqref{eq:def_q} we have
\begin{align}
p\bar{B}+(1-p)T\tilde{\mu}
&=\big(p+(1-p)T\big)\cdot\mathbb{E}\big[\min(E^{(i)},\,E_c)\big]\nonumber\\*
&=\big(p+(1-p)T\big)\cdot\big((1-p)\tilde{\mu}+pE_c\big),\label{eq:lambda0_identity_lemma}
\end{align}
hence the first term in~\eqref{eq:lambda0_lower_bound} is equal to
$\tfrac{p}{T}\tfrac{1}{1+(1-p)\tilde{\mu}+pE_c}$.
Next, the second term can be written as
\begin{align}
\tfrac{p-p'}{T}\tfrac{T-1}{T-1+TE_c-\bar{B}}
&=\tfrac{p-p'}{T}\tfrac{1}{1+E_c-\tfrac{\bar{B}-E_c}{T-1}}\nonumber\\
&=\tfrac{p-p'}{T}\tfrac{1}{1+(1-p)\tilde{\mu}+pE_c},\label{eq:lambda0_second_identity}
\end{align}
where the second equality follows from~\eqref{eq:def_Ec}.
Substituting in~\eqref{eq:lambda0_lower_bound}, we conclude
\[
\lambda_0 > \tfrac{p'}{T}\tfrac{1}{1+(1-p)\tilde{\mu}+pE_c} > 0.
\]

We continue with the derivative with respect to $\beta$:
\begin{align}
\left.\frac{\partial\mathscr{L}}{\partial\beta}\right|_{\gamma^\star,\beta^\star,\beta^\star,\beta_1^\star}
&=\tfrac{1-p}{T}\tfrac{p+(1-p)T}{p+(1-p)T+p\bar{B}+(1-p)T\tilde{\mu}}\nonumber\\*
&\quad+\tfrac{T-1}{T}\sum_{E_c<x\leq\bar{B}}\tfrac{P_E(x)}{T-1+Tx-\bar{B}}
+\lambda_0+\nu\nonumber\\
&=\tfrac{1}{T}\tfrac{p+(1-p)T}{p+(1-p)T+p\bar{B}+(1-p)T\tilde{\mu}}+\nu,
\end{align}
where the second equality is due to~\eqref{eq:lambda0}.
It follows that
\begin{equation}
\nu=-\tfrac{1}{T}\tfrac{p+(1-p)T}{p+(1-p)T+p\bar{B}+(1-p)T\tilde{\mu}}.
\label{eq:nu}
\end{equation}

The derivative with respect to $\beta_0$ is given by:
\begin{align}
\left.\frac{\partial\mathscr{L}}{\partial\beta_0}\right|_{\gamma^\star,\beta^\star,\beta_0^\star,\beta_1^\star}
&=\tfrac{p+(1-p)T}{T}\tfrac{-(1-p)}{p+(1-p)T+p\bar{B}+(1-p)T\tilde{\mu}}\nonumber\\*
&\quad\ -(1-p)\nu,
\end{align}
which is equal to zero due to~\eqref{eq:nu}.

Finally, we differentiate with respect to $\beta_1(x)$ for $E_c<x\leq\bar{B}$:
\begin{align}
\left.\frac{\partial\mathscr{L}}{\partial\beta_1(x)}\right|_{\gamma^\star,\beta^\star,\beta_0^\star,\beta_1^\star}
&=-\tfrac{T-1}{T}\tfrac{P_E(x)}{T-1+Tx-\bar{B}}-\lambda_1(x)-\nu P_E(x).\nonumber
\end{align}
Equating to zero gives
\begin{align}
\lambda_1(x)&=P_E(x)\left(-\nu-\tfrac{1}{T}\tfrac{T-1}{T-1+Tx-\bar{B}}\right)\nonumber\\*
&\overset{\text{(i)}}{>}P_E(x)\left(\tfrac{1}{T}\tfrac{p+(1-p)T}{p+(1-p)T+p\bar{B}+(1-p)T\tilde{\mu}}
-\tfrac{1}{T}\tfrac{T-1}{T-1+TE_c-\bar{B}}\right)\nonumber\\
&\overset{\text{(ii)}}{=}P_E(x)\left(\tfrac{1}{T}\tfrac{1}{1+(1-p)\tilde{\mu}+pE_c}
-\tfrac{1}{T}\tfrac{T-1}{T-1+TE_c-\bar{B}}\right)\nonumber\\
&\overset{\text{(iii)}}{=}0,
\end{align}
where (i) is due to~\eqref{eq:nu} and because $x>E_c$,
(ii) is by~\eqref{eq:lambda0_identity_lemma},
and (iii) is by \eqref{eq:lambda0_second_identity}.\qed

\bibliographystyle{IEEEtran}
\bibliography{energy_harvesting}

\begin{thebibliography}{10}
\providecommand{\url}[1]{#1}
\csname url@samestyle\endcsname
\providecommand{\newblock}{\relax}
\providecommand{\bibinfo}[2]{#2}
\providecommand{\BIBentrySTDinterwordspacing}{\spaceskip=0pt\relax}
\providecommand{\BIBentryALTinterwordstretchfactor}{4}
\providecommand{\BIBentryALTinterwordspacing}{\spaceskip=\fontdimen2\font plus
\BIBentryALTinterwordstretchfactor\fontdimen3\font minus
  \fontdimen4\font\relax}
\providecommand{\BIBforeignlanguage}[2]{{%
\expandafter\ifx\csname l@#1\endcsname\relax
\typeout{** WARNING: IEEEtran.bst: No hyphenation pattern has been}%
\typeout{** loaded for the language `#1'. Using the pattern for}%
\typeout{** the default language instead.}%
\else
\language=\csname l@#1\endcsname
\fi
#2}}
\providecommand{\BIBdecl}{\relax}
\BIBdecl

\bibitem{BlockiidBernoulli}
D.~Shaviv and A.~{\"{O}}zg{\"{u}}r, ``Online power control for block {i.i.d.}
  {Bernoulli} energy harvesting channels,'' in \emph{IEEE Wireless Commun. and
  Netw. Conf. Workshops (WCNCW)}, Mar. 2017, pp. 1--6.

\bibitem{BlockiidGLOBECOM}
------, ``Online power control for block {i.i.d.} energy harvesting channels,''
  submitted to IEEE GLOBECOM, 2017.

\bibitem{YangUlukus2012}
J.~Yang and S.~Ulukus, ``Optimal packet scheduling in an energy harvesting
  communication system,'' \emph{IEEE Trans. Commun.}, vol.~60, no.~1, pp.
  220--230, Jan. 2012.

\bibitem{TutuncuogluYener2012}
K.~Tutuncuoglu and A.~Yener, ``Optimum transmission policies for battery
  limited energy harvesting nodes,'' \emph{IEEE Trans. Wireless Commun.},
  vol.~11, no.~3, pp. 1180--1189, Mar. 2012.

\bibitem{Ozeletal2011}
O.~Ozel, K.~Tutuncuoglu, J.~Yang, S.~Ulukus, and A.~Yener, ``Transmission with
  energy harvesting nodes in fading wireless channels: Optimal policies,''
  \emph{IEEE J. Sel. Areas Commun.}, vol.~29, no.~8, pp. 1732--1743, Sep. 2011.

\bibitem{DP0}
M.~Zafer and E.~Modiano, ``Optimal rate control for delay-constrained data
  transmission over a wireless channel,'' \emph{IEEE Trans. Inf. Theory},
  vol.~54, no.~9, pp. 4020--4039, Sep. 2008.

\bibitem{DP1}
C.~K. Ho and R.~Zhang, ``Optimal energy allocation for wireless communications
  powered by energy harvesters,'' in \emph{Proc. IEEE Int. Symp. Information
  Theory (ISIT)}, Jun. 2010, pp. 2368--2372.

\bibitem{DP3}
P.~Blasco, D.~Gunduz, and M.~Dohler, ``A learning theoretic approach to energy
  harvesting communication system optimization,'' \emph{IEEE Trans. Wireless
  Commun.}, vol.~12, no.~4, pp. 1872--1882, Apr. 2013.

\bibitem{ho2012optimal}
C.~K. Ho and R.~Zhang, ``Optimal energy allocation for wireless communications
  with energy harvesting constraints,'' \emph{IEEE Trans. Signal Process.},
  vol.~60, no.~9, pp. 4808--4818, Sep. 2012.

\bibitem{Mitran}
M.~B. Khuzani and P.~Mitran, ``On online energy harvesting in multiple access
  communication systems,'' \emph{IEEE Trans. Inf. Theory}, vol.~60, no.~3, pp.
  1883--1898, Mar. 2014.

\bibitem{online_infiniteB2}
V.~Sharma, U.~Mukherji, V.~Joseph, and S.~Gupta, ``Optimal energy management
  policies for energy harvesting sensor nodes,'' \emph{IEEE Trans. Wireless
  Commun.}, vol.~9, no.~4, pp. 1326--1336, Apr. 2010.

\bibitem{online_infiniteB3}
R.~Rajesh, V.~Sharma, and P.~Viswanath, ``Capacity of fading {Gaussian} channel
  with an energy harvesting sensor node,'' in \emph{Proc. IEEE Global
  Telecommun. Conf. (GLOBECOM)}, Dec. 2011, pp. 1--6.

\bibitem{online_infiniteB4}
R.~Srivastava and C.~E. Koksal, ``Basic performance limits and tradeoffs in
  energy-harvesting sensor nodes with finite data and energy storage,''
  \emph{IEEE/ACM Trans. Netw. (TON)}, vol.~21, no.~4, pp. 1049--1062, Aug.
  2013.

\bibitem{info2013}
Q.~Wang and M.~Liu, ``When simplicity meets optimality: Efficient transmission
  power control with stochastic energy harvesting,'' in \emph{Proc. IEEE
  INFOCOM}, Apr. 2013, pp. 580--584.

\bibitem{amirnavaei2015online}
F.~Amirnavaei and M.~Dong, ``Online power control strategy for wireless
  transmission with energy harvesting,'' in \emph{Proc. IEEE Int. Workshop
  Signal Process. Advances in Wireless Commun. (SPAWC)}, Jun. 2015, pp. 6--10.

\bibitem{xu2014throughput}
J.~Xu and R.~Zhang, ``Throughput optimal policies for energy harvesting
  wireless transmitters with non-ideal circuit power,'' \emph{IEEE J. Sel.
  Areas Commun.}, vol.~32, no.~2, pp. 322--332, Feb. 2014.

\bibitem{Kazerouni2015}
A.~Kazerouni and A.~{\"{O}}zg{\"{u}}r, ``Optimal online strategies for an
  energy harvesting system with {Bernoulli} energy recharges,'' in \emph{Proc.
  Int. Symp. Modeling and Optimization in Mobile, Ad Hoc, and Wireless Netw.
  (WiOpt)}, May 2015, pp. 235--242.

\bibitem{DongOzgur2014}
Y.~Dong and A.~{\"{O}}zg{\"{u}}r, ``Approximate capacity of energy harvesting
  communication with finite battery,'' in \emph{Proc. IEEE Int. Symp.
  Information Theory (ISIT)}, Jun./Jul. 2014, pp. 801--805.

\bibitem{DongFarniaOzgur2015}
Y.~Dong, F.~Farnia, and A.~{\"{O}}zg{\"{u}}r, ``Near optimal energy control and
  approximate capacity of energy harvesting communication,'' \emph{IEEE J. Sel.
  Areas Commun.}, vol.~33, no.~3, pp. 540--557, Mar. 2015.

\bibitem{NearOptimal1}
D.~Shaviv and A.~{\"{O}}zg{\"{u}}r, ``Universally near optimal online power
  control for energy harvesting nodes,'' \emph{IEEE Journal in Selected Areas
  in Communication}, vol.~34, no.~12, pp. 3620--3631, Dec. 2016.

\bibitem{NearOptimal2}
------, ``Universally near-optimal online power control for energy harvesting
  nodes,'' in \emph{Proc. IEEE Int. Conf. Commun. (ICC)}, May 2016, pp. 1--6.

\bibitem{HuseyinWiOpt}
H.~A. Inan and A.~{\"{O}}zg{\"{u}}r, ``Online power control for the energy
  harvesting multiple access channel,'' in \emph{Proc. Int. Symp. Modeling and
  Optimization in Mobile, Ad Hoc, and Wireless Netw. (WiOpt)}, May 2016, pp.
  1--6.

\bibitem{HuseyinISIT}
H.~A. Inan, D.~Shaviv, and A.~{\"{O}}zg{\"{u}}r, ``Capacity of the energy
  harvesting {Gaussian} {MAC},'' in \emph{Proc. IEEE Int. Symp. Inf. Theory
  (ISIT)}, Jul. 2016, pp. 2744--2748.

\bibitem{BakninaUlukusISIT2016MAC}
A.~Baknina and S.~Ulukus, ``Online policies for multiple access channel with
  common energy harvesting source,'' in \emph{Proc. IEEE Int. Symp. Inf. Theory
  (ISIT)}, Jul. 2016, pp. 2739--2743.

\bibitem{BakninaUlukusISIT2016Broadcast}
------, ``Online scheduling for energy harvesting broadcast channels with
  finite battery,'' in \emph{Proc. IEEE Int. Symp. Inf. Theory (ISIT)}, Jul.
  2016, pp. 1984--1988.

\bibitem{BakninaUlukusJSAC2016}
------, ``Optimal and near-optimal online strategies for energy harvesting
  broadcast channels,'' \emph{IEEE J. Sel. Areas Commun.}, vol.~34, no.~12, pp.
  3696--3708, Dec. 2016.

\bibitem{VaranYenerWCNC2017}
B.~Varan and A.~Yener, ``Online transmission policies for cognitive radio
  networks with energy harvesting secondary users,'' presented at the IEEE
  Wireless Commun. and Netw. Conf. (WCNC), Mar. 2017.

\bibitem{MaoHassibiarXiv}
\BIBentryALTinterwordspacing
W.~Mao and B.~Hassibi. (2016) Capacity analysis of discrete energy harvesting
  channels. [Online]. Available: \url{https://arxiv.org/abs/1606.08973}
\BIBentrySTDinterwordspacing

\bibitem{Davidsbook}
D.~Tse and P.~Viswanath, \emph{Fundamentals of wireless communication}.\hskip
  1em plus 0.5em minus 0.4em\relax Cambridge, U.K: Cambridge Univ. Press, 2005.

\bibitem{ShavivNguyenOzgur2016}
D.~Shaviv, P.-M. Nguyen, and A.~{\"{O}}zg{\"{u}}r, ``Capacity of the
  energy-harvesting channel with a finite battery,'' \emph{IEEE Trans. Inf.
  Theory}, vol.~62, no.~11, pp. 6436--6458, Nov. 2016.

\bibitem{Bertsekas2001vol1}
D.~P. Bertsekas, \emph{Dynamic Programming and Optimal Control}, 3rd~ed.\hskip
  1em plus 0.5em minus 0.4em\relax Belmont, MA: Athena Scientific, 2001,
  vol.~1.

\bibitem{Bertsekas2001vol2}
------, \emph{Dynamic Programming and Optimal Control}, 4th~ed.\hskip 1em plus
  0.5em minus 0.4em\relax Belmont, MA: Athena Scientific, 2001, vol.~2.

\bibitem{OzelUlukus2012}
O.~Ozel and S.~Ulukus, ``Achieving {AWGN} capacity under stochastic energy
  harvesting,'' \emph{IEEE Trans. Inf. Theory}, vol.~58, no.~10, pp.
  6471--6483, Oct. 2012.

\bibitem{Puterman2005}
M.~L. Puterman, \emph{Markov Decision Processes: Discrete Stochastic Dynamic
  Programming}.\hskip 1em plus 0.5em minus 0.4em\relax Hoboken, NJ: John Wiley
  \& Sons, 2005.

\bibitem{asmussen2008applied}
S.~Asmussen, \emph{Applied probability and queues}.\hskip 1em plus 0.5em minus
  0.4em\relax New York, NY: Springer, 2008, vol.~51.

\bibitem{ross2014introduction}
S.~M. Ross, \emph{Introduction to probability models}.\hskip 1em plus 0.5em
  minus 0.4em\relax Academic Press, 2014.

\bibitem{durrett2010probability}
R.~Durrett, \emph{Probability: theory and examples}.\hskip 1em plus 0.5em minus
  0.4em\relax Cambridge, U.K: Cambridge Univ. Press, 2010.

\bibitem{Jafar2006}
S.~Jafar, ``Capacity with causal and noncausal side information: A unified
  view,'' \emph{IEEE Trans. Inf. Theory}, vol.~52, no.~12, pp. 5468--5474, Dec.
  2006.

\end{thebibliography}

\end{document}